\documentclass[11pt,oneside]{amsart}
\usepackage{amssymb}
\usepackage{color}
\usepackage{amsmath}
\usepackage{latexsym}
\usepackage{verbatim}
\usepackage[right=1in,top=1in,bottom=1in]{geometry}
\usepackage[linktocpage=true]{hyperref}

\hypersetup{ colorlinks   = true, 
urlcolor  = blue, 
linkcolor    = black, 
citecolor   = black 
}

\def\t{{\theta}}
\def\l{{\lambda}}

\def\d{{\delta}}
\def\b{{\beta}}
\def\a{{\alpha}}
\def\e{{\varepsilon}}
\def\beq{\begin{equation}}
\def\eeq{\end{equation}}

\definecolor{purple}{rgb}{.5,0,1}

\newcommand{\Z}{{\mathbb Z}}
\newcommand{\R}{{\mathbb R}}

\newcommand{\C}{{\mathbb C}}
\newcommand{\D}{{\mathbb D}}

\newcommand{\PP}{{\mathbb P}}

\newcommand{\bbS}{{\mathbb S}}
\newcommand{\bbI}{{\mathbb I}}
\newcommand{\pD}{{\partial\mathbb{D}}}

\newcommand{\CA}{{\mathcal A}}
\newcommand{\CB}{{\mathcal B}}
\newcommand{\CC}{{\mathcal C}}
\newcommand{\CD}{{\mathcal D}}
\newcommand{\CE}{{\mathcal E}}
\newcommand{\CF}{{\mathcal F}}
\newcommand{\CG}{{\mathcal G}}

\newcommand{\CI}{{\mathcal I}}
\newcommand{\CJ}{{\mathcal J}}

\newcommand{\CL}{{\mathcal L}}
\newcommand{\CM}{{\mathcal M}}

\newcommand{\CR}{{\mathcal R}}

\newcommand{\CU}{{\mathcal U}}

\newcommand{\dd}{{\mathrm{d}}}

\newcommand{\SL}{{\mathrm{SL}}}
\newcommand{\GL}{{\mathrm{GL}}}

\newcommand{\SU}{{\mathrm{SU}}}

\renewcommand{\Im}{{\mathrm{Im}}}

\newcommand{\set}[1]{\left\{#1\right\}}
\newcommand{\norm}[1]{\left\Vert#1\right\Vert}
\newcommand{\vabs}[1]{\left\vert#1\right\vert}

\DeclareMathOperator{\supp}{{supp}}

\newcommand{\sprod}[2]{\left<#1,#2\right>}

\newcommand{\base}{K}
\newcommand{\lbase}{K}

\newtheorem{theorem}{Theorem}[section]

\theoremstyle{definition}
\newtheorem*{remark*}{Remark}
\theoremstyle{definition}
\newtheorem{remark}[theorem]{Remark}
\theoremstyle{definition}
\newtheorem{defi}[theorem]{Definition}

\newtheorem{lemma}[theorem]{Lemma}
\newtheorem{prop}[theorem]{Proposition}

\newtheorem{corollary}[theorem]{Corollary}

\newtheorem{claim}{Claim}

\numberwithin{theorem}{section}
\numberwithin{equation}{section}

\begin{document}

\author[V.\ Bucaj]{Valmir Bucaj}
\address{Department of Mathematics, Rice University, Houston, TX~77005, USA}
\email{valmir.bucaj@rice.edu}
\thanks{V.B., D.D., V.G., T.V.\ were supported in part by NSF grant DMS--1361625.}

\author[D.\ Damanik]{David Damanik}
\address{Department of Mathematics, Rice University, Houston, TX~77005, USA}
\email{damanik@rice.edu}

\author[J.\ Fillman]{Jake Fillman}
\address{Department of Mathematics, Virginia Polytechnic Institute and State University, 225 Stanger Street -- 0123, Blacksburg, VA~24061, USA}
\email{fillman@vt.edu}
\thanks{J.F.\ was supported in part by an AMS-Simons travel grant, 2016--2018}

\author[V.\ Gerbuz]{Vitaly Gerbuz}
\address{Department of Mathematics, Rice University, Houston, TX~77005, USA}
\email{vitaly.gerbuz@rice.edu}

\author[T.\ VandenBoom]{Tom VandenBoom}
\address{Department of Mathematics, Rice University, Houston, TX~77005, USA}
\email{tv4@rice.edu}

\author[F.\ Wang]{Fengpeng Wang}
\address{School of Mathematical Sciences, Ocean University of China, Qingdao, China 266100 and Department of Mathematics, Rice University, Houston, TX~77005, USA}
\email{fw13@rice.edu}
\thanks{F.W.\ was supported by CSC (No.201606330003) and NSFC (No.11571327).}

\author[Z.\ Zhang]{Zhenghe Zhang}
\address{Department of Mathematics, Rice University, Houston, TX~77005, USA}
\email{zzhang@rice.edu}
\thanks{Z.Z.\ was supported in part by an AMS-Simons travel grant, 2014--2016}

\title[Localization for the 1D Anderson model]{Localization for the One-Dimensional Anderson model via Positivity and Large Deviations for\\ the Lyapunov Exponent}

\begin{abstract}
We provide a complete and self-contained proof of spectral and dynamical localization for the one-dimensional Anderson model, starting from the positivity of the Lyapunov exponent provided by F\"urstenberg's theorem. That is, a Schr\"odinger operator in $\ell^2(\Z)$ whose potential is given by independent identically distributed (i.i.d.) random variables almost surely has pure point spectrum with exponentially decaying eigenfunctions and its unitary group exhibits exponential off-diagonal decay, uniformly in time.  This is achieved by way of a new result: for the Anderson model, one typically has Lyapunov behavior for all generalized eigenfunctions. We also explain how to obtain analogous statements for extended CMV matrices whose Verblunsky coefficients are i.i.d., as well as for half-line analogs of these models.
\end{abstract}

\maketitle

\setcounter{tocdepth}{1}
\tableofcontents

\section{Introduction}

\subsection{The Goal in a Nutshell}

This paper is centered around the following fundamental result:

\begin{theorem}[spectral localization for the 1D Anderson model]\label{t.main}
Consider the family $\{H_\omega\}_{\omega \in \Omega}$ of random Schr\"odinger operators, acting in $\ell^2(\Z)$ via
$$
[H_\omega \psi](n) = \psi(n+1) + \psi(n-1) + V_\omega(n) \psi(n),
$$
where the potential $V_\omega$ is given by independent identically distributed random variables. It is assumed that the common distribution has a compact support that contains at least two elements. Then, almost surely, $H_\omega$ is spectrally localized, that is, it has pure point spectrum with exponentially decaying eigenfunctions.
\end{theorem}

We will provide a complete and relatively elementary derivation of this result, starting from the classical F\"urstenberg theorem about products of random matrices. In particular, as opposed to all previously published proofs of this result, we will not appeal to multi-scale analysis (MSA) as a black box.  The advantages in so doing extend beyond a mere simplification of the proof: in fact, our methods establish improve on previous asymptotic estimates and can be applied to prove spectral localization in other one-dimensional models which were previously inaccessible.

\subsection{What This Paper Accomplishes} \label{ssec:accomplish}


The key realization in this paper relates Lyapunov behavior and the existence of polynomially bounded solutions to the time-independent Schr\"odinger equation
\begin{equation}\label{eq:deve}
u(n+1) + u(n-1) + V_\omega(n) u(n) = E u(n).
\end{equation}
This difference equation admits a two-dimensional solution space, as any two consecutive values of $u$ determine all other values. Fixing $(u(0),u(-1))^\top$ as the point of reference, the linear map taking this vector to $(u(n),u(n-1))^\top$ is given by the so-called transfer matrix $M_n^E(\omega)$. Ergodicity of the full shift implies that for each $E$, there are $L(E) \ge 0 $ and $\Omega^E_-, \Omega^E_+ \subseteq \Omega$ with $\mu(\Omega^E_-) = \mu(\Omega^E_+) = 1$ such that
$$
L(E) = \begin{cases} \lim_{n \to \infty} \frac{1}{n} \log \| M_n^E(\omega) \| & \text{ for } \omega \in \Omega^E_+, \\ \lim_{n \to - \infty} \frac{1}{|n|} \log \| M_n^E(\omega) \| & \text{ for } \omega \in \Omega^E_-. \end{cases}
$$
The number $L(E)$ or the function $L(\cdot)$ are called the \emph{Lyapunov exponent}. F\"urstenberg's theorem implies that, in fact, $L(E) > 0$ for every $E$. Now, if $\omega \in \Omega^E_+$ (resp., $\omega \in \Omega^E_-$), then due to a result of Osceledec, there is a one-dimensional subspace of the solution space in which every element decays exponentially at $\infty$ (resp., $- \infty$), while every linearly independent solution grows exponentially at $\infty$ (resp., $- \infty$). More precisely, the rate is given by the Lyapunov exponent, that is, the decaying solutions obey
$$
\lim_{n \to \pm \infty} \frac{1}{|n|} \log (|u(n)|^2 + |u(n-1)|^2)^{1/2} = - L(E),
$$
while growing solutions obey
$$
\lim_{n \to \pm \infty} \frac{1}{|n|} \log (|u(n)^2 + |u(n-1)|^2)^{1/2} = L(E).
$$
These decay/growth statements are sometimes referred to as \emph{Lyapunov behavior}.

Let $\CG(H_\omega)$ denote the set of energies $E$ for which the difference equation \eqref{eq:deve} admits a non-trivial solution $u$ satisfying a linear upper bound,
\begin{equation}\label{eq:polyBounds}
|u(n)| \le C_u (1 + |n|)
\end{equation}
with $C_u$ a $u$-dependent constant. Energies $E$ in $\mathcal{G}(H_\omega)$ are called \emph{generalized eigenvalues} of $H_\omega$, and the corresponding linearly bounded solutions $u$ are called \emph{generalized eigenfunctions}. It is a classical theorem of Schnol \cite{Schnol1981} and Simon \cite{Simon1981JFA} that
\[
\mathcal G (H_\omega) \subseteq \sigma(H_\omega)
\text{ and }
\chi_{\R \setminus \mathcal G (H_\omega)}(H_\omega)
=0.
\]
In particular, $\mathcal G (H_\omega)$ supports all spectral measures of $H_\omega$.

Our proof of Theorem \ref{t.main} evades MSA by proving that every generalized eigenfunction exhibits Lyapunov behavior:

\begin{theorem}
\label{t:LyapConvGenEvals}
For $\mu$-almost every $\omega\in \Omega$ and every $E \in \mathcal G(H_\omega)$, one has
\[
\lim_{n\to\infty} \frac{1}{n}\log\|M_n^E(\omega)\|
=
\lim_{n \to -\infty} \frac{1}{|n|} \log\| M_{n}^E(\omega) \|
=
L(E).
\]
\end{theorem}

Our Theorem~\ref{t:LyapConvGenEvals} is precise with regard to the proof of spectral localization in a subtle way.  By Oseledec's theorem, for every $E$ we have Lyapunov behavior at both $\pm \infty$ for almost every $\omega$. Suppose we could turn this around and claim for almost every $\omega$ Lyapunov behavior at both $\pm \infty$ for \emph{every} $E$. Let us call this statement the \emph{Anderson localizer's dream}. Then spectral localization would immediately follow! Indeed, fix an $\omega$ from this full measure set. Since spectrally almost all $E$'s admit a polynomially bounded solution, and each solution is either exponentially increasing or decreasing at $\infty$ and either exponentially increasing or decreasing at $-\infty$, it follows that all polynomially bounded solutions must decay exponentially (at the Lyapunov rate) at both $\pm \infty$, and hence be genuine eigenfunctions. Thus, spectrally all energies are genuine eigenvalues and spectral localization (pure point spectrum with exponentially decaying eigenfunctions for almost all $\omega$'s) follows.

Alas, things aren't so easy. First of all, switching the order of the quantifiers only gives Lyapunov behavior at both $\pm \infty$ for almost every $\omega$ and (for example) Lebesgue almost every $E$, as a consequence of Fubini's theorem (applied to the product measure $\mu \times \mathrm{Leb}$). As briefly alluded to above, this is already sufficient to establish spectral localization thanks to spectral averaging if the single-site distribution has an absolutely continuous component. But it is decidedly not sufficient in the singular case. Second of all, and more importantly, the Anderson localizer's dream is even known to fail! Namely, Gorodetski and Kleptsyn have shown in \cite{GK17} that for almost all $\omega$'s, Lyapunov behavior fails for energies $E$ from a dense $G_\delta$ subset of the spectrum $\Sigma$ in the sense that
$$
0 = \liminf_{n \to \infty} \frac{1}{n} \log \| M_n^E(\omega) \| < \limsup_{n \to \infty} \frac{1}{n} \log \| M_n^E(\omega) \| = L(E)
$$
(and a similar statement on the left half line).\footnote{As we were completing the work on this paper we learned that \cite{GK17} also contains a new proof of spectral localization for the 1D Anderson model, which arises as a byproduct of their extension of the classical F\"urstenberg theorem.} In other words, the Anderson localizer's dream is a mirage that should not be chased.

However, in our Theorem~\ref{t:LyapConvGenEvals}, we have found the appropriate modification of the Anderson localizer's dream: the desired Lyapunov behavior holds for all generalized eigenvalues!  This is precisely the set of energies for which one needs Lyapunov behavior to be able to deduce spectral localization; in particular, Theorem~\ref{t:LyapConvGenEvals} implies Theorem~\ref{t.main}.

Notice also that one gets as a natural byproduct that the decay rate of the eigenfunctions is given by the Lyapunov exponent, which is certainly to be expected. However, the existing proofs of spectral localization for the Bernoulli Anderson model \cite{CKM87, SVW98} merely prove exponential decay without an attempt to optimize the decay rate.\footnote{As pointed out on p.46 of \cite{CKM87}, the decay rate of the eigenfunctions obtained via their MSA approach will be at least $\frac12 L(E)$, but no further possible improvements are discussed.} To summarize, in this paper we have found the natural statement and what we believe to be the natural proof of the phenomenon of spectral localization for the one-dimensional Anderson model.

\medskip

We are also able to establish exponential dynamical localization,
$$
\sup\limits_{t \in \R}
|\sprod{\delta_n}{e^{-itH_\omega}\delta_m}|
\lesssim e^{\epsilon |m|} e^{-\beta|n-m|};
$$
see Theorem~\ref{t:DL} below for the precise statement. Prior related works are \cite{DS01, GD98, GK01, GK04, GK06}. Here, $\epsilon$ is arbitrarily small and $\beta$ is arbitrarily close to the best possible decay rate -- given by the minimum of the Lyapunov exponent on the almost sure spectrum. Thus, both of our spectral and dynamical localization results are established with the correct decay rate.

Furthermore, there are models for which similar localization results are expected to hold, but for which no MSA exists, and hence there is in fact no known localization result. To illustrate this point we consider the class of CMV matrices.

A CMV matrix arises in the representation of the map $f(z) \mapsto zf(z)$ in $L^2(\partial \D, \dd\mu)$ relative to a suitable basis, where $\mu$ denotes a probability measure on the unit circle that does not admit a finite support. It is a five-diagonal semi-infinite matrix that is determined by a sequence of \emph{Verblunsky coefficients} $\{\alpha_n\}_{n \in \Z_+} \subset \D$, which arise as the recursion coefficients of the orthogonal polynomials associated with $\mu$, and the derived quantities $\rho_n = \left( 1 - |\alpha_n|^2 \right)^{1/2}$:
\begin{equation} \label{def:cmv}
\small
\mathcal C
=
\begin{bmatrix}
\overline{\alpha_0} & \overline{\alpha_1}\rho_0 & \rho_1\rho_0 &&& & \\
\rho_0 & -\overline{\alpha_1}\alpha_0 & -\rho_1 \alpha_0 &&& & \\
& \overline{\alpha_2}\rho_1 & -\overline{\alpha_2}\alpha_1 & \overline{\alpha_3} \rho_2 & \rho_3\rho_2 & & \\
& \rho_2\rho_1 & -\rho_2\alpha_1 & -\overline{\alpha_3}\alpha_2 & -\rho_3\alpha_2 &  &  \\
&&& \overline{\alpha_4} \rho_3 & -\overline{\alpha_4}\alpha_3 & \overline{\alpha_5}\rho_4 & \rho_5\rho_4 \\
&&& \rho_4\rho_3 & -\rho_4\alpha_3 & -\overline{\alpha_5}\alpha_4 & -\rho_5 \alpha_4  \\
&&&& \ddots & \ddots &  \ddots
\end{bmatrix}.
\end{equation}
This matrix defines a unitary operator in $\ell^2(\Z_+)$, and the spectral measure corresponding to $\mathcal C$ and the vector $\delta_0$ is given by $\mu$.

This sets up a one-to-one correspondence between measures $\mu$ and coefficient sequences $\{\alpha_n\}_{n \in \Z_+}$, which has been extensively studied in recent years, mainly due to the infusion of ideas from Simon's monographs \cite{S1, S2}.

Similarly, an \emph{extended CMV matrix} is a unitary operator on $\ell^2(\Z)$ defined by a bi-infinite sequence $\{\alpha_n\}_{n \in \Z} \subset \D$ in an analogous way:
\begin{equation} \label{def:extcmv}
\small
\mathcal E
=
\begin{bmatrix}
\ddots & \ddots & \ddots &&&&&  \\
\overline{\alpha_0}\rho_{-1} & -\overline{\alpha_0}\alpha_{-1} & \overline{\alpha_1}\rho_0 & \rho_1\rho_0 &&& & \\
\rho_0\rho_{-1} & -\rho_0\alpha_{-1} & -\overline{\alpha_1}\alpha_0 & -\rho_1 \alpha_0 &&& & \\
&  & \overline{\alpha_2}\rho_1 & -\overline{\alpha_2}\alpha_1 & \overline{\alpha_3} \rho_2 & \rho_3\rho_2 & & \\
& & \rho_2\rho_1 & -\rho_2\alpha_1 & -\overline{\alpha_3}\alpha_2 & -\rho_3\alpha_2 &  &  \\
& &&& \overline{\alpha_4} \rho_3 & -\overline{\alpha_4}\alpha_3 & \overline{\alpha_5}\rho_4 & \rho_5\rho_4 \\
& &&& \rho_4\rho_3 & -\rho_4\alpha_3 & -\overline{\alpha_5}\alpha_4 & -\rho_5 \alpha_4  \\
& &&&& \ddots & \ddots &  \ddots
\end{bmatrix}.
\end{equation}
From the point of view of orthogonal polynomials, the study of $\mathcal C$ is more natural, but when the Verblunsky coefficients are generated by an invertible ergodic map (such as for example the full shift of primary interest in this paper), the study of $\mathcal E$ is more natural.

CMV matrices have received a lot of attention in recent years, due to the close analogy with Jacobi matrices, which has led to a plethora of results for CMV matrices that should be regarded as the proper analog of existing Jacobi matrix results. However, the case of \emph{random} CMV matrices is not as well understood as the case of random Jacobi matrices. Specifically, there are no known analogs of the Kunz-Souillard method or of MSA. Of course, by what was said above, this has left the Bernoulli case entirely inaccessible. On the other hand, our approach carries over to the case of CMV matrices and hence we are able to prove the desired localization result for the CMV Bernoulli case.

Thus, as in the Schr\"odinger case we fix a single-site distribution that is compactly supported inside the open unit disk $\D$ (rather than $\R$ as above). This induces random sequences $\{\alpha_n(\omega)\}_{n \in \Z_+}$ and $\{\alpha_n(\omega)\}_{n \in \Z}$, as well as random CMV matrices $\mathcal C_\omega$ and random extended CMV matrices $\mathcal E_\omega$.

The following theorem is the CMV analog of Theorem~\ref{t.main}.

\begin{theorem}[spectral localization for random extended CMV matrices]\label{t:ALCMV}
Consider the family $\{\CE_\omega\}_{\omega \in \Omega}$ of random extended CMV matrices, acting in $\ell^2(\Z)$, where the Verblunsky coefficents are  given by independent identically distributed random variables. It is assumed that the topological support of the common distribution is a compact subset of $\D$ that contains at least two elements. Then, almost surely, $\CE_\omega$ is spectrally localized, that is, it has pure point spectrum with exponentially decaying eigenfunctions. Moreover, the rate of decay at energy $z$ is exactly $L(z)$.
\end{theorem}

Due to the special interest to the orthogonal polynomial community, we also state explicitly the half-line version of the previous result:

\begin{theorem}[spectral localization for random CMV matrices]\label{t:ALCMVhalfline}
Consider the family $\{\CC_\omega\}_{\omega \in \Omega}$ of random CMV matrices, acting in $\ell^2(\Z_+)$, where the Verblunsky coefficents are  given by independent identically distributed random variables. It is assumed that the topological support of the common distribution is a compact subset of $\D$ that contains at least two elements. Then, almost surely, $\CC_\omega$ is spectrally localized, that is, it has pure point spectrum with exponentially decaying eigenfunctions.
\end{theorem}

As mentioned above, most approaches to localization for one-dimensional Schr\"odinger operators have not been carried over to CMV matrices yet, and hence localization for CMV matrices was not known in the expected generality. Theorem~\ref{t:ALCMVhalfline} remedies this deficit and establishes localization in the expected generality. As above, one could ask about the case where the single site distribution is not compactly supported in $\D$ and our method should extend to these cases under the appropriate assumptions (ensuring, e.g., the existence of the Lyapunov exponent, which is absolutely fundamental to this approach). Theorem~\ref{t:ALCMVhalfline} was previously known only in situations where spectral averaging was applicable: the single site distribution was assumed to be absolutely continuous with respect to either Lebesgue measure on $\D$ or arc length on a circle centered at the origin and of radius smaller than one; compare \cite[Theorem~12.6.3]{S2} and \cite{T92}. Simon and Teplyaev did not consider extended CMV matrices, and hence Theorem~\ref{t:ALCMV} is technically speaking new even in the case of absolutely continuous single-site distributions. However, our main point here is that singular distributions cannot be handled using spectral averaging techniques, and in particular the Bernoulli case (when the support of the single-site distribution has cardinality two) requires the tools developed in this paper. Moreover, as in the Schr\"odinger case, our method also yields suitable results concerning exponential dynamical localization, see Theorem~\ref{t:DLCMV} below for the precise statement.

There is some related work on another class of random unitary operators \cite{HJS06, J05}, but their results also require the absolute continuity of the single-site distribution.

For the sake of completeness, we mention that our approach provides a half-line version of Theorem~\ref{t.main} as well.

\subsection{Background and Context}\label{ssec:bg}

Let us describe the context and the relevance of Theorem~\ref{t.main}. The Nobel Prize winning work of Philip Warren Anderson suggested that randomness leads to localization of quantum states in suitable energy regions that depend on the strength of the randomness. A particular signature of such a localization effect is a spectral localization statement, which asserts that the spectral type of the associated Schr\"odinger operator is pure point in suitable energy regions, and the eigenfunctions corresponding to the eigenvalues in these energy regions decay exponentially.

For the sake of concreteness, let us consider the standard Anderson model, which is just the $d$-dimensional generalization of the operator family considered in Theorem~\ref{t.main}. That is, given a probability measure $\widetilde \mu$ on $\R$ whose topological support is compact and contains at least two points, we consider the product space $\Omega = (\mathrm{supp} \, \widetilde \mu)^{\Z^d}$ and the product measure $\mu = {\widetilde \mu}^{\Z^d}$. For every $\omega \in \Omega$ and $n \in \Z^d$, we set $V_\omega(n) = \omega_n$. This defines, for $\omega \in \Omega$, a potential $V_\omega : \Z^d \to \R$, and in turn a Schr\"odinger operator
$$
[H_\omega \psi](n) = \sum_{|m-n|_1 = 1} \psi(m) + V_\omega(n) \psi(n)
$$
in $\ell^2(\Z^d)$. Standard ergodicity arguments show that the spectrum and the spectral type of $H_\omega$ are almost surely independent of $\omega$, that is, there exist sets $\Sigma, \Sigma_\mathrm{pp}, \Sigma_\mathrm{sc}, \Sigma_\mathrm{ac}$, and a set $\Omega_0 \subseteq \Omega$ of full $\mu$-measure such that for every $\omega \in \Omega_0$, we have $\sigma(H_\omega) = \Sigma$ and $\sigma_\bullet(H_\omega) = \Sigma_\bullet$, $\bullet \in \{ \mathrm{pp}, \mathrm{sc}, \mathrm{ac} \}$. It is not too hard to show that
$$
\Sigma = [-2d,2d] + \mathrm{supp} \, \widetilde \mu.
$$

The assumption that $\mathrm{supp} \, \widetilde \mu$ be compact ensures that these operators are bounded; the real-valuedness of the potential ensures that they are also self-adjoint. The boundedness is not crucial, and one could in fact consider probability measures $\widetilde \mu$ with unbounded support. However, the phenomenon of Anderson localization already occurs in the bounded case, and many authors limit their attention to this case -- as do we. Furthermore, the assumption that $\mathrm{supp} \, \widetilde \mu$ contain more than one point excludes the trivial case of a constant potential, for which the Anderson localization phenomenon is obviously impossible (the spectrum is purely absolutely continuous in this case).

The spectral signature of Anderson localization is now the following. There exists a set $\Sigma_\mathrm{AL} \subseteq \Sigma$, which is a finite union of non-degenerate intervals, so that $\Sigma_\mathrm{AL} \subseteq \Sigma_\mathrm{pp}$ (or, really, ``$=$'') and $\mathrm{int} \, \Sigma_\mathrm{AL} \cap \Sigma_\mathrm{sc} = \mathrm{int} \, \Sigma_\mathrm{AL} \cap \Sigma_\mathrm{ac} = \emptyset$. This means that almost surely the spectrum of $H_\omega$ is pure point on $\Sigma_\mathrm{AL}$, and hence $H_\omega$ has a set of eigenvalues that is dense in $\Sigma_\mathrm{AL}$. The additional feature is that the associated eigenfunctions decay exponentially. The size of $\Sigma_\mathrm{AL}$ relative to the size of $\Sigma$ depends on the dimension and the strength of the randomness. It is expected that $\Sigma_\mathrm{AL} = \Sigma$ for $d=1$ and $d=2$, and that in general we only have $\Sigma_\mathrm{AL} \subseteq \Sigma$ for $d \ge 3$. However, $\Sigma_\mathrm{AL} = \Sigma$ does hold in the case $d \ge 3$ when the randomness is strong enough. More specifically, each of the connected components of $\Sigma_\mathrm{AL}$ is a neighborhood of a boundary point of $\Sigma$. The length of these intervals grows with increasing randomness, up to a point where they cover all of $\Sigma$. From this strength of randomness onwards, we have $\Sigma_\mathrm{AL} = \Sigma$.

Notice that Theorem~\ref{t.main} is precisely the expected statement $\Sigma_\mathrm{AL} = \Sigma$ for the case $d=1$. The localization part of the expected statement for the case $d \ge 3$ is known for sufficiently regular $\widetilde \mu$ as well, but proving that $\overline{\Sigma \setminus \Sigma_\mathrm{AL}} = \Sigma_\mathrm{ac} \not= \emptyset$ for sufficiently small randomness in the case $d \ge 3$ is the main open problem in the study of random Schr\"odinger operators. The expected statement for $d=2$ is not known. The best known result in the case $d = 2$ is the same as the best known result for $d \ge 3$, but no better. That is, one only knows localization (under suitable assumptions on $\widetilde \mu$) in neighborhoods of the boundary points of $\Sigma$, but not in all of $\Sigma$, as is expected. This problem (``prove spectral localization for the two-dimensional Anderson model throughout the whole spectrum for any strength of randomness'') is the second main open problem in the field. The third main open problem is to establish the localization result in dimensions $d \ge 2$, which as pointed out above is known under suitable assumptions on $\widetilde \mu$, for \emph{any} single-site measure $\widetilde \mu$. For example, the case where $\widetilde \mu$ has a non-zero pure point part is not covered by known localization results in higher dimensions yet.

Proofs of spectral localization results come in two general flavors. Some of them use methods that are strictly one-dimensional, and others work in any dimension. Among the strictly one-dimensional proofs we mention the Kunz-Souillard approach \cite{DG16, DKS83, KS80} and the approach via spectral averaging \cite{S85, SW86}. Proofs that work in arbitrary dimension include those based on multi-scale analysis (MSA) \cite{DK89, EK16, FMSS85, FS83, GK01} and the fractional moment method (FMM) \cite{A94, AM93, ASFH01}. However, most of these approaches are limited in terms of the single-site distributions to which they apply. Except for the MSA approach, all others require $\widetilde \mu$ to have a non-trivial absolutely continuous component, or to even be purely absolutely continuous. This leaves MSA as the only method available in cases where $\widetilde \mu$ is purely singular. The special case where $\widetilde \mu$ is supported on precisely two points is the hardest; this is commonly referred to as the Bernoulli case, and the operator family in this case is called the Bernoulli Anderson model. Unfortunately, the approach based on MSA is by far the most complex and difficult among the available approaches. This makes the treatment of the Bernoulli case complex and difficult, even in one dimension!

On the other hand, the case of one dimension is special in that F\"urstenberg's theorem about products of random matrices provides compelling evidence that spectral localization holds for all single-site distributions in this case. It shows that for all $\widetilde \mu$, and throughout $\Sigma$, solutions of the generalized eigenvalue equation have a strong tendency to be either exponentially increasing or exponentially decreasing. Coupled with the general fact that all spectral measures are supported on the set of energies that admit polynomially bounded solutions, this should imply that all such polynomially bounded solutions are in fact exponentially decreasing, and hence are eigenfunctions, and hence the spectrum is pure point because the spectral measures are supported on such energies, which turn out to be eigenvalues by the argument above. The starting point of this argument, the exponential behavior of solutions, has no analog in higher dimensions, and this is the reason why spectral localization is known at all energies in one dimension, is not known in two dimensions, and is in fact expected to fail in general for dimensions greater than two.

Alas, the argument outlined in the previous paragraph has a flaw which has to do with exceptional sets and uncountable unions of zero-measure sets. Localization proofs in one dimension that are based on the output of F\"urstenberg's theorem (i.e., all proofs different from those using the Kunz-Souillard method) must address this flaw. That is, they do implement the general strategy, but they address the complications that arise when uncountable unions of exceptional sets of zero measure are taken.

When $\widetilde \mu$ has a non-trivial absolutely continuous component, this is taken care of in a very elegant way by spectral averaging. One of the fundamental properties of spectral averaging (namely that the average of spectral measures turns out to be Lebesgue measure) allows one to simply ignore sets of zero Lebesgue measure, and this shows in effect that the flaw is a non-issue in this case.

On the other hand, when $\widetilde \mu$ is singular, the only option up to this point has been to somehow verify the assumptions that are necessary to start a MSA throughout the spectrum, which consists of proving a Wegner estimate and establishing an initial length scale estimate. The former is very difficult to establish in the Bernoulli case, and the latter follows from the positive Lyapunov exponents provided by F\"urstenberg's theorem. Thus, the work necessary to deal with the one-dimensional case in full generality (i.e., including the Bernoulli case) mainly focused on establishing a Wegner estimate. This was accomplished, by different methods, in two papers: by Carmona, Klein and Martinelli in 1987 \cite{CKM87} and by Shubin, Vakilian and Wolff in 1998 \cite{SVW98}. Once all the ingredients are in place, the MSA machine produces the desired spectral localization statement, and hence Theorem~\ref{t.main}.

However, this way of proving Theorem~\ref{t.main} is somewhat unsatisfactory. MSA is an inductive scheme that produces, with large probability, exponential off-diagonal decay statements for Green's functions associated with finite-volume restrictions of the random operators for a sequence of interval lengths. That is, one uses the positivity of the Lyapunov exponent to verify the initial length scale estimate, only to then work hard to inductively prove exponential decay statements that are directly related to, and should in fact follow from, the positivity of the Lyapunov exponent! What one really ought to do is to make full use of what positive Lyapunov exponents actually provide.

It has therefore been a well-recognized problem in the random operator community to find a more direct way of going from the positive Lyapunov exponents provided by F\"urstenberg's theorem to the spectral localization statement contained in Theorem~\ref{t.main}, that is, to find a one-dimensional proof of this one-dimensional result. This is precisely what we accomplish in this paper.

\subsection{Strategy of the Proof}

The basic strategy of our localization proof follows the ideas from Bourgain and Schlag's localization proof for half-line Schr\"odinger operators whose potentials are generated by the doubling map on the circle \cite{bs}. In particular, the main ingredients of this approach are uniform positivity of the Lyapunov exponent (LE), a uniform large deviation theorem (LDT) for the same, and a version of a lemma regarding the \emph{elimination of double resonances}. Here, uniformity is with respect to the spectral parameter $E$.

\medskip

The remainder of the paper is organized as follows. In Section~\ref{s:positivity_continuity_LE}, we show continuity and uniform positivity of the Lyapunov exponent as function of $E$, starting with verifying the conditions of F\"urstenberg's theorem for $\mathrm{SL}(2,\R)$, formulated as Theorem~\ref{furst}. This section is largely expository and details are provided for the convenience of the reader.

In Section~\ref{sec:LDT}, starting again with the conditions of F\"urstenberg's theorem, we supply a relatively simple proof of a uniform LDT for the Lyapunov exponent. Such results are known for i.i.d.\ random matrices \cite{lP1982} and one-parameter families of i.i.d.\ matrices \cite{Tsay1999}. However, we wish to emphasize that this proof of the LDT is new. It makes the use of the independence of the underlying dynamics more transparent and may have the potential to be useful elsewhere. In fact, by the author of \cite{Z16}, a version of the LDT is currently being worked out for the model in \cite{Z16} where the potentials are generated by some strongly mixing dynamics.  Some of the arguments in Section~\ref{sec:LDT} may be directly used there.

In Section~\ref{sec:holder}, we supply a simple proof of H\"older continuity of the LE and the integrated density of states. This result is well-known (compare \cite[Th\'eor\`eme~3]{lP1989AIHP}), but the modern proof via the Avalanche Principle is simpler.

In Section~\ref{sec:green}, we obtain suitable upper bounds on the norms of the transfer matrices and relate them to estimates on the Green functions of finite-volume truncations of the full-line model.

In Section~\ref{sec:localization}, we complete the proof of Theorems~\ref{t.main}, \ref{t:LyapConvGenEvals}, and Theorem~\ref{t:DL}; the starting point is Proposition~\ref{prop:doubleRes} which is an appropriate formulation of the elimination of double resonances in this setting.

We wish to emphasize that the tools we develop in Section 5 and 6 are more general than the similar ones in \cite{bs}, and this eventually enables us to show exponential dynamical localization for the Bernoulli-Anderson model for the first time.

Finally, in Section~\ref{sec:CMV}, following the arguments from the proof of Theorems~\ref{t.main} and \ref{t:DL}, we prove Theorems~\ref{t:ALCMV} and \ref{t:ALCMVhalfline}, and the CMV version of an exponential dynamical localization result, Theorem~ \ref{t:DLCMV}.

\section{Positivity and Continuity of the Lyapunov Exponent} \label{s:positivity_continuity_LE}

In this section, we will introduce several classical results concerning positivity and continuity of the Lyapunov exponent for products of i.i.d.\ random $\SL(2,\R)$ matrices, which will be instrumental in our proof of Anderson Localization.

Consider a probability space $(\CA,\widetilde\mu)$, and let $(\Omega,T,\mu)$ be the full shift space generated by $(\CA,\widetilde\mu)$. In other words, $\Omega=\CA^{\Z}$, $\mu=\widetilde\mu^{\Z}$ and
\[
(T\omega)_n=\omega_{n+1},
\quad
\omega \in \Omega, \; n \in \Z.
\]
Consider a map $M:\CA\rightarrow \SL(2,\R)$. For simplicity, we assume this map is bounded, that is
$$
\sup_{\alpha \in \CA} \|M(\alpha) \|
<
\infty.
$$
This generates a map $\Omega \rightarrow \SL(2,\R)$, which we also denote by $M$, via
$$
M(\omega)
=
M(\omega_0),
\quad
\omega \in \Omega,
$$
which in turn induces an $\SL(2,\R)$-cocycle over $T$ in a canonical way:
\begin{equation} \label{eq:cocycle}
(T,M): \Omega\times\R^2 \rightarrow \Omega\times\R^2,\ (T,M)(\omega,\vec v)=(T\omega, M(\omega)\vec v).
\end{equation}
We define the iterates of $M$ over the skew product by $(T,M)^n = (T^n,M_n)$ for $n \in \Z$. One can check that
$$
M_n(\omega)
=
\begin{cases}
 M(T^{n-1}\omega)\cdots M(\omega) & n > 0, \\
 \bbI & n = 0, \\
 \left[ M_{-n} (T^n\omega)\right]^{-1} & n < 0.
 \end{cases}
$$
We can (and do) view $M_n$ as the product of $n$ i.i.d.\ random $\SL(2,\R)$ matrices with common distribution $\widetilde \mu$ when $n \in \Z_+$.

The Lyapunov exponent of the cocycle \eqref{eq:cocycle} is defined by
$$
L
=
L(T,M)
:=
\lim_{n\rightarrow\infty}\frac1n\int_{\Omega} \! \log\|M_n\| \, \dd\mu
=
\inf_{n \ge 1} \frac1n\int_{\Omega} \! \log\|M_n\| \, \dd\mu\ge 0.
$$
By Kingman's Subadditive Ergodic Theorem, one also has
$$
L
=
\lim_{n\rightarrow\infty}\frac1n \log\|M_n(\omega)\|
$$
for $\mu$-almost every $\omega \in \Omega$.

In the current setting, it is convenient to consider the Lyapunov exponent as a function of the probability measure on $\SL(2,\R)$. Concretely, through the probability space $(\CA,\widetilde\mu)$ and the map $M$, we obtain a probability measure on $\SL(2,\R)$ via $\nu = M_*\widetilde \mu$, that is, the push-forward measure of $\widetilde \mu$ under the map $M$. Thus in our setting, we also sometime write the Lyapunov exponent as $L = L(M_*\widetilde\mu)$.

We denote by $\R\PP^1$ the real projective line, that is, $\R\PP^1$ is the set of lines in $\R^2$ that pass through the origin. It clear that each $M\in\SL(2,\R)$ induces a map on $\R\PP^1$, which will again be denoted by $M$. We say that a subgroup $G \subset \SL(2,\R)$ is \emph{strongly irreducible} if there is no finite non-empty set $\CF \subset \R\PP^1$ such that $M(\CF) = \CF$ for all $M \in G$.

 The following (special case of a) deep theorem of F\"urstenberg is essential for our analysis.

\begin{theorem}[F\"urstenberg \mbox{\cite[Theorem~8.6]{f}}]\label{furst}
Let $\nu$ be a probability measure on $\SL(2,\R)$ that satisfies
$$
\int \! \log \| M \| \, \dd\nu(M) < \infty.
$$
Denote by $G_\nu$ the smallest closed subgroup of $\SL(2,\R)$ that contains $\supp \, \nu$.

Assume
\begin{itemize}
\item[(i)]  $G_\nu$ is not compact.
\item[(ii)] $G_\nu$ is strongly irreducible.
\end{itemize}

Then, $L > 0$.
\end{theorem}
\begin{remark*}
Under condition (i), strong irreducibility of $G_\nu$ is equivalent to:
\begin{itemize}
\item[(ii')] There is no set $\CF \subseteq \R\PP^1$ of cardinality $1$ or $2$ such that $M(\CF) = \CF$ for all $M \in G_\nu$.
\end{itemize}
\end{remark*}

From this, one can immediately deduce global positivity of the Lyapunov exponent for the Anderson model as soon as the single-site distribution has at least two points in its support. We now precisely define the Anderson model and set up our notation.

\begin{defi}
Suppose that our probability space $(\CA, \widetilde \mu)$ consists of a compact set of real numbers; that is, we assume henceforth that
\[
\CA
=
\supp\widetilde\mu
\subset
\R
\]
and that $\CA$ is compact. Then, for each $\omega \in \Omega$,
\[
[H_\omega \psi](n)
=
\psi(n-1) + \psi(n+1) + \omega_n\psi(n),
\quad
n \in \Z, \; \psi \in \ell^2(\Z),
\]
defines a bounded self-adjoint operator on $\ell^2(\Z)$. For each $E \in \R$, we define the map $M^E:\CA \to \SL(2,\R)$ via
\[
M^E(\alpha)
=
\begin{bmatrix}
E - \alpha & -1 \\ 1 & 0
\end{bmatrix},
\quad
\alpha \in \CA,
\]
which may be extended to $\Omega$ as above. Then it is straightforward to verify that $M^E_n(\omega)$ is the $n$-step transfer matrix of the eigenvalue equation $H_\omega\phi = E\phi$. More specifically, $H_\omega \phi = E \phi$ if and only if
\[
\begin{bmatrix}
\phi(n) \\ \phi(n-1)
\end{bmatrix}
=
M_n^E(\omega)
\begin{bmatrix}
\phi(0) \\ \phi(-1)
\end{bmatrix}
\text{ for all } n \in \Z.
\]
The induced measure on $\SL(2,\R)$ will be denoted by $\nu_E = M_*^E \widetilde\mu$, and the Lyapunov exponent at energy $E$ is then defined and denoted by $L(E) := L(\nu_E)$.
\end{defi}

Defining
\begin{equation} \label{eq:almostSureSpecDef}
\Sigma
:=
\CA + [-2,2]
=
\set{x+y : x \in \CA, \; y \in [-2,2]},
\end{equation}
one can check that $\sigma(H_\omega) = \Sigma$ for $\mu$-a.e.\ $\omega \in \Omega$, e.g.\ by using generalized eigenfunctions.

Clearly, if $\#\CA = 1$, then the theory is quite trivial. Concretely, if $\CA$ consists of the single point $a \in \R$, then $\Omega$ contains only the constant sequence $\omega_n \equiv a$; in this case $\sigma(H_\omega) = \Sigma = [a-2,a+2]$, and the spectral type is purely absolutely continuous. Henceforth, we adopt the standing nontriviality assumption that $\#\CA \geq 2$.

\begin{theorem}\label{t:positivityLE}
In the Anderson model, $\nu_E = M^E_* \widetilde \mu$ satisfies assumptions {\rm(}i{\rm)} and {\rm(}ii{\rm)} from Theorem~\ref{furst} for every $E \in \R$. In particular, we have $L(E) > 0$ for every $E \in \R$.
\end{theorem}

\begin{proof}
Fix $E \in \R$. Since the support of the single-site distribution has cardinality at least two, it follows that $\nu_E$ also has at least two points in its support. Thus, $G_{\nu_E}$ contains at least two distinct elements of the form
$$
M_x
=
\begin{bmatrix}
 x & -1 \\ 1 & 0
\end{bmatrix},
$$
say, $M_a$ and $M_b$ with $a \neq b$. Note that
$$
A
=
M_a M_b^{-1}
=
\begin{bmatrix}
1 & a-b \\ 0 & 1
\end{bmatrix}
\in G_{\nu_E}.
$$
Taking powers of the matrix $A$, we see that $G_{\nu_E}$ is not compact.

Now, consider $V_1 := \mathrm{span}(\vec e_1)$, the projection of $\vec e_1 := (1,0)^\top$ to $\R\PP^1$. Then $A V_1 = V_1$ and, for every $V \in \R\PP^1$, $A^n V$ converges to $V_1$. Thus, if there is a nonempty finite invariant set of directions $\CF \subseteq \R\PP^1$, one must have $\CF = \{V_1\}$. However, we also have
$$
A'
=
M_a^{-1} M_b = \begin{bmatrix}1 & 0 \\ a-b & 1 \end{bmatrix}  \in G_\nu
$$
and $A' V_1 \neq V_1$. Thus, $G_{\nu_E}$ is strongly irreducible, so conditions (i) and (ii) of Theorem~\ref{furst} are met. Consequently, $L(E) = L(\nu_E) > 0$ by F\"urstenberg's theorem.
\end{proof}

The main application of positive Lyapunov exponents is that one obtains precise asymptotic statements about orbits under cocycle iterates. The following deterministic theorem supplies what we need.

\begin{theorem}[Ruelle \cite{Ruelle1979}]
\label{t:ruelle}
Suppose $A^{(n)} \in \SL(2,\R)$ obey
$$
\lim_{n \to \infty} \frac{1}{n} \log \| A^{(n)} \| = 0
$$
and
$$
\lim_{n \to \infty} \frac{1}{n} \log \| A^{(n)} \cdots A^{(1)} \| = L > 0.
$$
Then there exists a one-dimensional subspace $V \in \R\PP^1$ such that
$$
\lim_{n \to \infty} \frac{1}{n} \log \| A^{(n)} \cdots A^{(1)} \vec v \|
=
\begin{cases}
-L  &  \vec v \in V \setminus \{0\}\\
L   &  \vec v \in \R^2 \setminus V.
\end{cases}
$$
\end{theorem}

For a proof of Theorem~\ref{t:ruelle}, see \cite{Ruelle1979} or \cite[Theorem~2.8]{Dam:Random}.
\medskip

For the Anderson model, we will also want to know that $L(E)$ is a continuous function of $E$, which follows easily from a theorem that goes back to F\"urstenberg and Kifer in the 1980s. Given a sequence of Borel probability measures $\{\nu_k\}$ supported in $\SL(2,\R)$, we say $\nu_k$ converges to $\nu$ \emph{weakly and boundedly} if
\begin{equation} \label{eq:boundedconv}
\int_{\|M\| \geq N} \! \log^+\|M\| \, \dd\nu_k(M)
+
\int_{\|M\| \geq N} \! \log^+\|M\| \, \dd\nu(M)
\to 0
\end{equation} as $N \to \infty$, uniformly in $k$ and
\[
\int \! f \, \dd\nu_k\to \int \! f \, \dd\nu
\]
for all $f\in C_c(\SL(2,\R),\C)$,  the space of continuous complex-valued functions on $\SL(2,\R)$ with compact support. In \eqref{eq:boundedconv}, we use $\log^+(x) = \max(\log x,0)$ to denote the positive part of $\log x$. We first state a (special case of a) theorem of F\"urstenberg and Kifer, which is another cornerstone for our analysis.

\begin{theorem}[F\"urstenberg, Kifer \mbox{\cite[Theorem~B]{fk}}] \label{t:furst2}
Let $\nu$ be a probability measure on $\SL(2,\R)$ for which
\[
\mathrm{Fix}(G_\nu)
:=
\set{V \in \R\PP^1 : MV = V \text{ for every } M \in G_\nu}
\]
contains at most one element. Then, if $\nu_k \to \nu$ weakly and boundedly, it holds that
\[
\lim_{k\to\infty}L(\nu_k)
=
L(\nu).
\]
\end{theorem}

Applying this theorem to the Anderson model, we obtain continuity of $L$ as a function of $E \in \R$.

\begin{theorem}\label{t:continuityLE}
In the Anderson model, $L(E)$ is continuous as a function of $E$. In particular, $L$ is uniformly positive in the sense that
\begin{equation}\label{eq:UPLE}
\gamma
:=
\inf_{E \in \R} L(E)> 0.
\end{equation}
\end{theorem}

\begin{proof}
Let $E \in \R$ be given, and, as above, let $\nu_E = M^E_*\widetilde\mu$ denote the induced measure on $\SL(2,\R)$. By Theorem~\ref{t:positivityLE}, $\nu_E$ is strongly irreducible; in particular, $\mathrm{Fix}(G_{\nu_E})$ is empty, so $\nu_E$ satisfies the assumption of Theorem~\ref{t:furst2}. Consequently, to prove continuity of $L(E)$ in $E$, it suffices to show that $\nu_{E_n} \to \nu_{E}$ weakly and boundedly whenever $E_n \to E$.

Given a sequence $E_n \to E$, one can verify that there is a uniform compact subset of $\SL(2,\R)$ that simultaneously supports $\nu_E$ and every $\nu_{E_n}$, so \eqref{eq:boundedconv} follows. Thus, we only need to show the ``weakly'' part. By dominated convergence, we obtain
$$
\lim_{n \to \infty}\int \! f \, \dd\nu_{E_n}
=
\lim_{n \to \infty} \int \! f\circ M^{E_n} \, \dd \widetilde\mu
=
\int \! f\circ M^{E} \, \dd\widetilde\mu
=
\int \! f \, \dd\nu_{E}
$$
for every $f\in C(\SL(2,\R),\C)$, which concludes the proof of continuity.

Combining continuity of $L$ with Theorem~\ref{t:positivityLE}, we see that $L(E)$ is uniformly bounded from below away from zero on any compact set. On the other hand, one can check that
$$
\lim_{|E|\to\infty}L(E) = +\infty,
$$
so \eqref{eq:UPLE} follows.
\end{proof}

\begin{remark}\label{r:LEconcave}
We have defined $\gamma$ to be the global minimum of $L(E)$ over $E \in \R$. However, it turns out that
\begin{equation} \label{eq:UPLEspec}
\gamma = \min_{E \in \Sigma} L(E),
\end{equation}
that is, $L$ achieves its minimum value on $\Sigma$. To see this, one may use the Thouless formula \cite{AvrSim1983DMJ, CraSim1983DMJ, Thouless1972JPhysC}:
\begin{equation} \label{eq:thouless}
L(E)
=
\int_\R \log|E - x| \, \dd N(x).
\end{equation}
In \eqref{eq:thouless}, $\dd N$ denotes the \emph{density of states measure} (DOS) associated with the family $\{H_\omega\}_{\omega\in\Omega}$, which is defined by
\[
\int \! g(E) \, \dd N(E)
=
\int_\Omega \! \langle \delta_0, g(H_\omega) \delta_0 \rangle \, \dd\mu(\omega)
\]
for bounded measurable functions $g$. The DOS is a Borel probability measure with $\supp \dd N = \Sigma$ \cite{AvrSim1983DMJ}. Then, if $E \notin \Sigma$, $L$ is differentiable at $E$ with
\[
\frac{\dd L}{\dd E}(E)
=
\int_\R \frac{\dd N(x)}{E-x}.
\]
This follows by using dominated convergence and noting that $E \notin \Sigma$ allows one to uniformly bound $h^{-1}(\log|E + h - x| - \log|E-x|)$ over $x \in \Sigma$ for sufficiently small $h$. Applying this argument again, one gets
\[
\frac{\dd^2 L}{\dd E^2}(E)
=
-\int_\R \frac{\dd N(x)}{(E-x)^2}
<
0.
\]
Then, since $L$ is continuous on $\R$, and $L(E) \to \infty$ as $|E| \to \infty$, \eqref{eq:UPLEspec} follows. The reader should note that the argument used to deduce differentiability of $L$ no longer works if $E \in \Sigma  = \supp\dd N$, so we only get smoothness of $L$ outside of $\Sigma$.
\end{remark}

The last piece of information that we will need to run our arguments is a statement to the effect that, for any sequence of unit vectors $v_n \in \R^2$, $\|M^E_n(\omega) v_n \|$ grows like $e^{nL(E)}$ as $n \to \infty$ ($\mu$-almost surely); this statement will enable us to prove an initial scale estimate that we can then use to inductively prove a Large Deviation Theorem. To do this, we must verify one final technical hypothesis on the group $G_{\nu_E}$.

\begin{defi}\label{d:contraction} Given a subset $G$ of $\mathrm{GL}(d, \mathbb{\R})$, we say that $G$ is \emph{contracting} if there exists a sequence $\{g_n\}_{n = 1}^\infty$ in $G$ for which $\displaystyle \|g_n\|^{-1}g_n$ converges to a rank-one operator.
\end{defi}

\begin{prop}\label{p:contraction}
In the Anderson model, $G_{\nu_E}$ is contracting for every $E$.
\end{prop}

\begin{proof}
Given $E$, let $a$, $b$, and $A$ be as in the proof of Theorem~\ref{t:positivityLE} , and take $g_n = A^n$. It is easy to check that
$$
\|g_n\|^{-1}g_n
\to
\begin{bmatrix} 0 & 1\\ 0 & 0\end{bmatrix},
$$
which concludes the proof.
\end{proof}

\begin{prop}
\label{prop:LEconvergentVects}
For any $E$ and any convergent sequence $\set{\vec v_n}_{n=1}^\infty$ of unit vectors in $\R^2$,
$$
\lim_{n\to\infty}\frac{1}{n}\log \left\| M_n^E (\omega)\vec v_n \right\|
=
L(E)
$$
for almost every $\omega\in\Omega$.
\end{prop}

\begin{proof}
This follows from \cite[pp.\ 53--54, Corollary~3.4]{bougerollacroix}, which may be invoked because $G_{\nu_E}$ is strongly irreducible and contracting.
\end{proof}

Proposition~\ref{prop:LEconvergentVects} represents the starting point of our new proof of the Large Deviation Theorem. We want to say a few words about this proposition since it is already a highly nontrivial result. First, Proposition~\ref{prop:LEconvergentVects} trivially holds true if the Lyapunov exponent is zero, that is, when $L=0$. Consequently, the nontrivial part of this proposition lies in the case when $L>0$.

From Oseledec's Multiplicative Ergodic Theorem, we know that if $L>0$, the cocycle \eqref{eq:cocycle} has two invariant sections (i.e.\ measurable maps $\Lambda^s,\Lambda^u:\Omega \to \R \PP^1$ with $A\Lambda^\bullet = \Lambda^\bullet \circ T$ for $\bullet \in \{s,u\}$) that are called \emph{stable} and \emph{unstable sections}. For almost every $\omega$, vectors drawn from the stable subspace $\Lambda^s(\omega)$ will contract exponentially fast in forward time (under the cocycle map) with the rate $L$, while vectors drawn from the unstable subspace $\Lambda^u(\omega)$ contract in backward time. Away from the stable direction, every vector grows exponentially in forward time with the rate $L$. Thus, Propsition~\ref{prop:LEconvergentVects} is a more general and sophisticated version of the following statement: for products of i.i.d.\ random $\SL(2,\R)$ matrices obeying the conditions of Theorem~\ref{furst} and ~\ref{t:furst2}, every nonzero vector in $\R^2$ is not in the stable direction of the phase $\omega$ with probability one.

\section{Large Deviation Estimates for Products of i.i.d.\ Matrices}
\label{sec:LDT}

For all $E\in\R$, Theorem~\ref{t:positivityLE} and Proposition~\ref{p:contraction} imply that $G = G_{\nu_E}$ is noncompact, strongly irreducible, and contracting. From Theorem~\ref{t:continuityLE}, we know that $L$ is a continuous function of $E$ (which is a consequence of continuity of the cocycle map and the contracting property). We will use these properties to deduce a suitable uniform (in $E$) Large Deviation Theorem (LDT).

Henceforth, define
\[
\hat{\Sigma}
:=
[-\kappa,\kappa],
\quad
\kappa := 2 + \max_{\alpha \in \CA}|\alpha|
\]
In particular,  $\hat\Sigma$ is a compact interval containing the almost sure spectrum, $\Sigma$, defined in \eqref{eq:almostSureSpecDef}. The goal of this section is to supply a simple proof of the following uniform LDT using strong irreducibility and contractivity of $G_{\nu_E}$.

\begin{theorem}\label{t:LDT}
For any $\e>0$, there exist $C = C(\e) > 0$, $\eta = \eta(\e) > 0$ such that
$$
\mu\left\{\omega \in \Omega: \left|\frac1n\log\left\|M_n^E(\omega) \right\| - L(E) \right| \geq \e\right\}
\leq
Ce^{-\eta n}
$$
for all $n \in \Z_+$ and all $E \in \hat{\Sigma}$.
\end{theorem}

We begin by proving a first step towards this estimate. Let us denote the set of all unit vectors in $\R^2$ by
\[
\bbS^1
=
\set{\vec v \in \R^2 : |v_1|^2 + |v_2|^2 = 1}.
\]

\begin{prop} \label{pr:ldtfirststep}
For any $\e > 0$, $0<\delta<1$, and $E \in \hat\Sigma$, there exists $N = N(\e,\delta,E)$ such that
\begin{equation} \label{eq:LDT:firststep}
\mu\left\{\omega:\left|\frac1n \log \left\| M_n^E(\omega) \vec v \right\| - L(E) \right| \geq \e \right\}
<
\delta,
\end{equation}
for all $n \geq N$ and $\vec v \in \bbS^1$.
\end{prop}

\begin{proof}
Suppose for the sake of establishing a contradiction that the statement of the proposition is false; that is, suppose that there exist $E_0 \in \hat \Sigma$, $\delta_0 \in (0,1)$, $\e_0 > 0$, a sequence of integers $n_k \to \infty$ and unit vectors $\vec v_{k} \in \bbS^1$ such that
\begin{equation} \label{eq:ldtfirststep:badWs}
\mu\! \set{ \omega : \left|\frac{1}{n_k}\log \left\| M_{n_k}^{E_0}(\omega)\vec v_k \right\| - L(E_0) \right| \geq \e_0 }
\geq
\delta_0
\text{ for every } k \in \Z_+.
\end{equation}
Since $E_0$ plays no role in the argument, we suppress it from the notation for the remainder of the proof. By passing to a subsequence, we may assume that the vectors $\vec v_{k}$ converge to a vector $\vec v_\infty \in \bbS^1$. Consequently, since $n_k^{-1} \log\|M_{n_k}\vec v_k\| - L$ is uniformly bounded on $\Omega$, we may apply Proposition~\ref{prop:LEconvergentVects} and dominated convergence to get
\begin{equation} \label{eq:ldtfirststep:lyapConv}
\lim_{k \to \infty} \int_\Omega \left|\frac{1}{n_k}\log\|M_{n_k}(\omega)\vec v_k\| - L\right| \, \dd\mu(\omega)
=
0.
\end{equation}

On the other hand, letting $A_k$ denote the set on the left-hand side of \eqref{eq:ldtfirststep:badWs}, we get
\[
\int_\Omega \left|\frac{1}{n_k}\log\|M_{n_k}\vec v_{k}\| - L\right| \, \dd\mu
\geq
\int_{A_k} \left|\frac{1}{n_k}\log\|M_{n_k}\vec v_{k}\| - L\right| \, \dd\mu
\geq
\e_0\d_0
\]
for every $k \in \Z_+$, which contradicts \eqref{eq:ldtfirststep:lyapConv}.
\end{proof}

Our next goal is to extend this proposition to a neighborhood of $E$, so we need to estimate how changes in $E$ perturb $M_n$. To do this, we define
\begin{equation} \label{eq:FnOmegaDef}
F_n(\omega,E)
=
\frac{1}{|n|}\log\|M_n^E(\omega)\|,
\quad
n \in \Z,  \; \omega \in \Omega, \; E \in \R,
\end{equation}
where we adopt the convention $F_0 \equiv 1$. Throughout the paper, there will be various uniform bounds that can be controlled in terms of bounds on the single-step matrices. Hence, we introduce
\begin{equation} \label{eq:oneStepTMBoundDef}
\Gamma
:=
\sup\set{
\left\| M^E(\alpha) \right\| : E \in \hat\Sigma, \; \alpha \in \CA}.
\end{equation}
Since $\Omega$ and $\hat \Sigma$ are compact and $M^E(\omega)$ is continuous as a function of $(E,\omega) \in \hat\Sigma \times \Omega$, it follows that $\Gamma$ is finite. From the definitions, it is easy to see that
\[
|F_n(\omega,E)|
\leq
\log\Gamma \text{ for all } n \in \Z, \; \omega \in \Omega, \; E \in \hat\Sigma.
\]
Combining this with \eqref{eq:UPLE}, we get the uniform bounds
\begin{equation} \label{eq:unifLEbounds}
\gamma
\leq
L(E)
\leq
\log\Gamma
\text{ for all } E \in \hat\Sigma.
\end{equation}

The following lemma is a straightforward calculation.

\begin{lemma} \label{lem:Fnlipschitz}
We have
\begin{align}
\label{eq:MnDiffBound}
\left\| M_n^E(\omega) - M_n^{E'}(\omega') \right\|
& \leq
n \Gamma^{n-1}\left(|E-E'|+ \max_{0\leq j < n}|\omega_j - \omega_j'| \right).
\end{align}
for all $E, E' \in \Sigma$, $\omega, \omega' \in \Omega$, and $n \in \Z_+$. Therefore,
\begin{align}
\label{eq:Fnderiv}
|F_n(\omega,E) - F_n(\omega',E')|
& \leq
\Gamma^{n-1}\left( |E-E'|+ \max_{0\leq j < n}|\omega_j - \omega_j'| \right)
\end{align}
for all $E$, $E'$, $\omega$, $\omega'$, and $n$. Furthermore, with
\[
L_n(E)
:=
\int_\Omega \! F_n(\omega,E) \, \dd\mu(\omega),
\]
one has
\begin{equation} \label{eq:LnLipsch}
|L_n(E) - L_n(E')|
\le
\Gamma^{n-1} |E-E'|.
\end{equation}
\end{lemma}

\begin{prop}
\label{pr:ldt1.5step}
For any $\e > 0$, $0 < \delta < 1$, and $E \in \hat\Sigma$, there exists $N' = N'(\e, \delta, E)$ such that, for every integer $n \geq N'$, there exists $\rho = \rho(n) > 0$ with the property that
\begin{equation} \label{eq:LDT:1.5step}
\mu\left\{\omega:\left|\frac1n\log \left\| M_n^{E'}(\omega)\vec v \right\| - L(E')\right| \geq \e \right\}
<
\delta
\end{equation}
for all unit vectors $\vec v \in \bbS^1$ whenever $|E-E'| < \rho$ and $E' \in \hat \Sigma$.
\end{prop}

\begin{proof}
Fix $\e > 0$, $0 < \delta < 1$, and $E \in \hat\Sigma$, and put $N' = N(\e/2, \delta, E)$ from Proposition \ref{pr:ldtfirststep}.  Let $n \geq N$ be given, and let $\rho = \rho(n)>0$ be chosen so that
\[
|L(E) - L(E')|
<
\e/4
\text{ whenever } E' \in \hat \Sigma \text{ and } |E - E'| < \rho,
\]
which can be done by continuity of the Lyapunov exponent and compactness of $\hat\Sigma$. If necessary, shrink $\rho$ to ensure that
\[
\rho \leq \frac{\e}{4\Gamma^{n-1}},
\]
where $\Gamma$ is as in \eqref{eq:oneStepTMBoundDef}. For any $\vec v \in \bbS^1$, we have
\begin{equation} \label{eq:ldt1.5step:pertContainment}
\set{\omega:
\left| \frac1n\log \|M_n^{E'}(\omega) \vec v \| - L(E') \right| \geq \e}
\subset
\set{\omega:
\left| \frac1n\log \|M_n^E(\omega) \vec v \| - L(E)\right| \geq \e/2}
\end{equation}
whenever $E' \in \hat \Sigma$ and $|E-E'| < \rho$. Concretely, if $\omega$ lies in the complement of the right-hand side of \eqref{eq:ldt1.5step:pertContainment}, we have
\begin{align*}
\left|\frac1n\log\|M_n^{E'}(\omega) \vec v\| - L(E')\right|
& \leq
\Gamma^{n-1}|E-E'| + \left|\frac1n\log\|M_n^E(\omega)\vec v \| - L(E')\right| \\
& \leq
\frac{\e}{4} + \left|\frac1n\log\|M_n^E(\omega)\vec v \| - L(E)\right| + \left|L(E) - L(E')\right| \\
& <
\varepsilon,
\end{align*}
where the first inequality uses Lemma~\ref{lem:Fnlipschitz}. Thus, \eqref{eq:ldt1.5step:pertContainment} holds and the conclusion of the proposition follows from Proposition~\ref{pr:ldtfirststep} and our choice of $N$.
\end{proof}

In what follows, we will need to use the following discrete Chebyshev-type inequality.

\begin{lemma} \label{l:large_average_arithm}
Let $P \in \Z_+$ and $a_1,\ldots,a_P \in \R$ be given. Given $L < \max_j a_j \leq B$ and $\delta > 0$, define
\[
\CJ_{L,\delta}
=
\set{j : a_j > L + \delta}.
\]
If
\begin{equation} \label{eq:average_greater}
\frac{1}{P} \sum_{j=1}^P a_j
\geq
L + \e,
\end{equation}
for some $\e > 0$, then
\begin{equation} \label{eq:avGrtr:CJcardBound}
\#\CJ_{L,\delta}
\geq
P \frac{\e - \delta}{B - L - \delta}
\end{equation}
for every $0 < \delta < B - L$.
\end{lemma}
\begin{proof}
Assume that \eqref{eq:average_greater} holds and let $m := \#\CJ_{L,\delta}$. Using \eqref{eq:average_greater} and splitting the sum over $\CJ_{L,\delta}$ and its complement, one obtains
\begin{equation*}
P(L+\e)
\leq
\sum\limits_{j = 1}^{P} a_j
=
\sum\limits_{j\in \CJ_{L,\delta}} a_j + \sum\limits_{j \notin \CJ_{L,\delta}} a_j
\le
mB + (P-m)(L+\delta).
\end{equation*}
Solving for $m$, we get \eqref{eq:avGrtr:CJcardBound}.
\end{proof}

 We may now combine our foregoing work to fashion the final stepping stone before proving the main LDT: a vectorwise uniform LDT. Compare \cite[Theorem~4]{Tsay1999}.
\begin{prop}\label{pr:almost_independence}
For every $\varepsilon>0$ there exist constants $C,\eta > 0$ such that
\begin{equation*}
\mu \set{\omega: \vabs{\frac{1}{n}\log \norm{M_n^E(\omega)\vec{v}} - L(E)} \geq \varepsilon} \leq Ce^{-\eta n}
\end{equation*}
for every $n \in \Z_+$, $\vec v \in \bbS^1$, and $E \in \hat\Sigma$.
\end{prop}

\begin{proof}
Let $\Gamma$ be defined as in \eqref{eq:oneStepTMBoundDef}, put $B = \log\Gamma + 1$, fix $\e \in (0,1)$, and notice that
\[
\e
<
B-L(E)
\text{ for every } E \in \hat\Sigma
\]
by \eqref{eq:unifLEbounds}. Motivated by Lemma~\ref{l:large_average_arithm}, we define $\xi > 0$ by
\[
\xi
:=
\frac{\e/4}{B-\gamma - \frac{\e}{4}},
\]
and observe that
\begin{equation} \label{eq:almostIndep:xidef}
\xi
\leq
\frac{\e/4}{B-L(E) - \frac{\e}{4}} \text{ for every } E \in \hat\Sigma
\end{equation}
by \eqref{eq:unifLEbounds}. In particular, $\xi$ depends on $\e$, but not on $E$. Now, fix $\delta > 0$ small enough that $1+\delta e \leq e^{\xi/2}$, and let $E \in \hat\Sigma$ be given. Then, put $N = N'(\e/4, \d; E)$ and $\rho = \rho(N)$ as in Proposition~\ref{pr:ldt1.5step}, and let $E' \in (E-\rho,E+\rho)$ and $n \in \Z_+$ be given. Writing $n = NP + r$ with $P\in\Z_+$ and $0 \le r < N$, one can check that
\begin{equation*}
\log\|M_n^{E'}(\omega)\vec v\|
=
\sum^{P-1}_{p=0} \log\left\| M_N^{E'}(T^{pN}\omega) \vec{v}_p \right\|
+ \log\left\| M_r^{E'}(T^{PN}\omega)\vec{v}_P \right\|,
\end{equation*}
where
$$
\vec v_p
=
\vec v_p(\omega,E')
:=
\frac{M_{pN}^{E'}(\omega)\vec v}{\|M_{pN}^{E'}(\omega)\vec v\|},
\quad
0 \le p \le P.
$$
We first deal with the set
\begin{equation*}
\CB^{+}_n
=
\CB^+_n(E',\e,\vec v)
:=
\left\{\omega:\frac1n\log\left\| M_n^{E'}(\omega)\vec v \right\| - L(E') \geq \e\right\}.
\end{equation*}
For $n$ large enough, one has
\begin{equation*}
\CB^{+}_n
\subseteq
\set{\omega: \frac{1}{P}\sum^{P-1}_{p=0}\frac{1}{N}\log\| M_N^{E'}(T^{pN}\omega) \vec{v_p}(\omega,E')\|
\geq
L(E') + \frac{\varepsilon}{2}}.
\end{equation*}
In light of this, Proposition~\ref{pr:ldt1.5step}, and our choice of $\xi$ in \eqref{eq:almostIndep:xidef}, we obtain
\begin{equation}\label{eq:B1_intersection}
\CB^{+}_n
\subseteq
\bigcup_{\substack{\CJ \subset [0,P-1]\cap \Z \\ \#\CJ \geq \xi P}}\;
\bigcap\limits_{p\in\CJ}
A_p,
\end{equation}
where $A_p = A_p(E',\e,\vec v)$ is given by
\[
A_p
:=
\set{\omega: \frac{1}{N}\log\|M_N^{E'}(T^{pN}\omega)\vec{v_p}(\omega,E')\| \geq L(E') + \frac{\varepsilon}{4}}.
\]
Thus, it remains to bound the measure of sets of the form
\[
A_\CJ = \bigcap_{p \in \CJ} A_p
\]
with $\CJ \subset [0,P-1]\cap\Z$ a set having cardinality at least $\xi P$. To that end, we notice that whether or not $\omega \in A_p$ depends only on the coordinates $(\omega_{0},\omega_1,\ldots,\omega_{(p+1)N-1})$. To capture this dependence, we introduce the following grouping of coordinates. Suppose $\#\CJ = m \geq \xi P$, write $\CJ = \set{p_1<p_2<\ldots<p_m}$, and define
\[
\vec \omega_j
:=
\left(\omega_{(p_{j-1} + 1)N} \ldots \omega_{(p_j + 1)N-1} \right),
\quad
1 \leq j \leq m,
\]
where we take $p_0 = -1$ by convention. Thus, we obtain the following grouping of the coordinates of $\omega$:
\[
\omega
=
(\ldots, \underbrace{\omega_0,\ldots,\omega_{(p_1+1)N-1}}_{\vec{\omega}_1\in\Omega_1 := \CA^{(p_1+1)N}},
\ldots,
\underbrace{\omega_{(p_{m-1}+1)N},\ldots,\omega_{(p_m+1)N-1}}_{\vec{\omega}_m\in\Omega_m := \CA^{(p_m-p_{m-1})N}},\ldots).
\]
Denoting the $(p_j-p_{j-1})N$-fold product of $\widetilde \mu$ with itself on $\Omega_{j}$ by $\widetilde \mu_j$, \eqref{eq:B1_intersection} gives us
\begin{align}
\nonumber
\mu \left(A_\CJ \right)
& =
\int\limits_{\Omega} \! \prod\limits_{p\in\CJ}\chi_{A_p}(\omega) \, \dd \mu(\omega)\\
\label{eq:intersection_estimate_step_1}
& =
\int\limits_{\Omega_1} \! \cdots \! \int\limits_{\Omega_m} \! \left[ \prod_{j=1}^m \chi_{A_{p_j}}(\vec{\omega}_1,\ldots,\vec\omega_j) \right] \,
\dd\widetilde\mu_m(\vec{\omega}_m) \,  \cdots \, \dd\mu_1(\vec{\omega}_1).
\end{align}
In the innermost integral the ``important variables'' that govern the growth of $\norm{M_N \vec v_{p_m}}$ correspond to the last $N$ coordinates of $\vec\omega_m$. So we further split $\vec\omega_m$ into its last $N$ entries and first $(p_m- p_{m-1}-1)N$ entries. Write $\vec\sigma = (\sigma_1,\ldots,\sigma_N)$ for the terminal $N$ coordinates of $\vec\omega_m$, and let $\vec\omega_m^1$ denote the remaining initial $(p_m-p_{m-1}-1)N$ coordinates of $\vec\omega_m$. Then, our goal is to estimate
\begin{equation}\label{eq:intersection_estimate_step_2}
\int\limits_{\Omega_m} \! \chi_{A_{p_m}}(\vec{\omega}_1,\ldots,\vec{\omega}_m) \, \dd\widetilde\mu_m(\vec{\omega}_m)
=
\int\limits_{\Omega_m^1} \int\limits_{\CA^N} \!
\chi_{A_{p_m}}(\vec{\omega}_1,\ldots,\vec{\omega}_m^1,\vec\sigma) \,
\dd\widetilde\mu^N(\vec\sigma) \, \dd\widetilde\mu^{(p_m-p_{m-1}-1)N}(\vec{\omega}_m^1),
\end{equation}
uniformly over $\vec\omega_1,\ldots,\vec\omega_{m-1}$. Note that in the definition of $A_{p_m}$, $M_N(T^{p_m N}\omega)$ depends only on $\vec\sigma$, and $\vec v_{p_m}$ depends only on $(\vec{\omega}_1,\vec{\omega}_2,\ldots,\vec{\omega}_m^1)$. Consequently,
\begin{align*}
\int\limits_{\CA^N} \! \chi_{A_{p_m}}(\vec{\omega}_1,\vec{\omega}_2,\ldots,\vec{\omega}_m^1,\vec\sigma) \, \dd\widetilde\mu^N(\vec \sigma)
=
\widetilde\mu^N\set{\vec\sigma: \frac{1}{N}\log\|M_N^{E'}(\vec\sigma)\vec{v}_{p_m} \| \geq L + \frac{\varepsilon}{4}}
<
\delta
\end{align*}
for our choice of $N$ by  Proposition~\ref{pr:ldt1.5step}. Now plugging back into \eqref{eq:intersection_estimate_step_2} gives us
\begin{equation*}
\int\limits_{\Omega_m} \! \chi_{A_{p_m}}(\vec{\omega}_1,\vec{\omega}_2,\ldots,\vec{\omega}_m) \, \dd\widetilde\mu_m(\vec{\omega}_m)
<
\delta.
\end{equation*}
Inductively applying the same argument $m$ times, we get
\begin{equation} \label{eq:ACJmeasBound}
\mu(A_\CJ) \leq\delta^m
\text{ whenever } \CJ \subset [0,P-1]\cap \Z \text{ and } \#\CJ = m.
\end{equation}
Bounding the measure of $\CB_n^{+}$ is now a matter of counting and our choice of $\delta$.  Namely, if we write $I = [0,P-1] \cap \Z$, \eqref{eq:B1_intersection} and \eqref{eq:ACJmeasBound} imply
\begin{align*}
\mu(\CB_n^{+})
\leq
\sum_{\substack{\CJ \subseteq I \\ \#\CJ \ge \xi P}} \delta^{\#\CJ}
\le
e^{-\xi P} \sum_{\CJ \subseteq I} (\delta e)^{\#\CJ}
=
e^{-\xi P}(1+\delta e)^P
\leq
e^{-\xi P/2}
\end{align*}
by our choice of $\delta$. Taking
\[
\eta_0
=
\eta_0(E,\e)
:=
\frac{\xi}{3N},
\]
we find that
\[
\mu(\CB_n^{+})
\leq
e^{-\xi P/2}
\leq
e^{-\eta_0 n}
\]
for all sufficiently large $n$ (note that the largeness condition depends solely on $E$ and $\e$). One bounds the $\mu$-measure of
\begin{equation*}
\CB^{-}_n
=
\CB^-_n(E',\e,\vec v)
:=
\left\{\omega:\frac1n\log\|M^{E'}_n(\omega)\vec v\|-L(E') \leq -\e\right\}.
\end{equation*}
similarly and obtains
\begin{equation*}
\mu(\CB^{-}_n) \leq  e^{-\eta_0 n}
\end{equation*}
for all sufficiently large $n$ by following the argument used to estimate $\mu(\CB_n^+)$. Naturally, this yields
\begin{equation} \label{eq:almostthere}
\mu \set{\omega: \vabs{\frac{1}{n}\log \norm{M^{E'}_n(\omega)\vec{v}} - L(E')} \geq \varepsilon}
=
\mu(\CB_n^+ \cup \CB_n^-)
\leq
2e^{-\eta_0 n}
\end{equation}
for $n$ sufficiently large.

Thus, for each $E$ in $\hat\Sigma$, we find $\eta_0 = \eta_0(E,\e)$, $\rho = \rho(E,\e)$, and $N'' = N''(E,\e)$, so that \eqref{eq:almostthere} holds for all $E' \in (E-\rho,E+\rho)$, $n \geq N''$, and $\vec v \in \bbS^1$.  Then, we obtain the conclusion of the theorem via a straightforward compactness argument (using compactness of $\hat\Sigma$).

\end{proof}

\begin{proof}[Proof of Theorem~\ref{t:LDT}]
Let $E \in \hat\Sigma$ and $\e > 0$ be given. Notice that
\begin{equation}\label{eq:main_abs_value_open}
\left\{\omega: \left|\frac1n\log\|M_n^E(\omega)\|-L\right| \geq \e\right\}
=
\CB_n^+ \cup \CB_n^-,
\end{equation}
where $\CB_n^\pm = \CB_n^\pm(E,\e)$ are given by
\[
\CB_n^\pm
:=
\set{\omega: \pm\left( \frac1n\log\|M_n^E(\omega)\|- L \right) \geq \e}.
\]
In terms of the sets from Proposition~\ref{pr:almost_independence}, one has
\[
\CB_n^+(E,\e) = \bigcup\limits_{\vec v \in \bbS^1} \CB_n^+(E,\e,\vec v),
\quad
\CB_n^-(E,\e) = \bigcap\limits_{\vec v \in \bbS^1} \CB_n^-(E,\e,\vec v).
\]
Thus, denoting the standard basis in $\R^2$ by $\{\vec e_1, \vec e_2\}$, we get $\CB^-_n(E,\e) \subseteq \CB_n^-(E,\e,\vec e_1)$, and hence
\[
\mu\left(\CB_n^-(E,\e)\right)
\leq
\mu\left(\CB_n^-(E,\e,\vec e_1)\right)
\leq
Ce^{-\eta_1 n}
\]
by Proposition~\ref{pr:almost_independence}.

It remains to estimate the measure of $\CB_n^+(E,\e)$. Since
\[
\norm{M_n(\omega)}
\leq
\sqrt{2} \max_{j = 1,2} \norm{M_n(\omega)\vec e_j},
\]
we get
\[
\CB_n^+(E,\e)
\subseteq
\CB_n^+(E,\e/2,\vec e_1) \cup \CB_n^+(E,\e/2,\vec e_2)
\]
for all $n$ large enough that $n\e \geq \log 2$, which in turn gives
\begin{equation*}
\mu(\CB_n^+(E,\e))
\leq
Ce^{-\eta_2 n} + Ce^{-\eta_3 n}.
\end{equation*}
Choosing $\eta = \min\set{\eta_1,\eta_2,\eta_3}$, the theorem follows.
\end{proof}

\section{H\"older Continuity of the Lyapunov Exponent}
\label{sec:holder}

The purpose of this section is to supply a simple proof of the H\"older continuity of the Lyapunov exponent of the Anderson model. In essence, once we have a uniform lower bound on the Lyapunov exponent as in \eqref{eq:UPLE} and a uniform LDT as in Theorem~\ref{t:LDT}, H\"older continuity of $L$ follows from the approach developed by Goldstein--Schlag \cite{goldstein}.

We will concentrate on the regularity of $L$ on the interval $\hat\Sigma$. By Remark~\ref{r:LEconcave}, $L$ is already a smooth function of $E$ away from $\hat\Sigma$.

\begin{theorem}\label{t:HolderContinuity}
There exist constants $C > 0, \beta > 0$ depending solely on $\widetilde{\mu}$ such that
\begin{equation} \label{eq:holder}
|L(E)-L(E')|
\leq
C|E-E'|^\beta
\end{equation}
for all $E,E' \in \hat\Sigma$.
\end{theorem}

\begin{remark*}
Another important object in the spectral analysis of the Anderson model is the accumulation function of the density of states measure, called the \emph{integrated density of states} (IDS):
\[
N(E)
=
\int_{(-\infty,E]} \!\dd N.
\]
By the Thouless formula, the IDS is (almost) the Hilbert transform of $L$; consequently, one can deduce quantitative continuity estimates for the IDS from such estimates on $L$. In particular, one can deduce H\"older continuity of $N$ as a function of $E$ from Theorem~\ref{t:HolderContinuity} with the same choice of $\beta$.
\end{remark*}

\begin{corollary}[{\cite[Th\`eor\`eme~3]{lP1989AIHP}}]
There exists a constant $C > 0$ such that
\[
|N(E) - N(E')|
\leq
C |E-E'|^\beta
\]
for all $E,E' \in \R$, where $\beta$ is as in Theorem~\ref{t:HolderContinuity}.
\end{corollary}

\begin{proof}
This follows from uniform positivity of $L$, Theorem~\ref{t:HolderContinuity}, and standard arguments using the Thouless Formula, \eqref{eq:thouless}. See, e.g., the proof of \cite[Theorem~6.1]{goldstein}, particularly the argument on pp.\ 175--176.
\end{proof}

The key ingredient in the approach of Goldstein--Schlag is the \emph{Avalanche Principle} \cite{goldstein}. Basically, the Avalanche Principle permits us good control on the norm of a product of $\SL(2,\R)$ matrices provided we have suitable estimates on consecutive pairwise products. A bit more precisely, if we consider a product like
\[
A
=
\prod^1_{j=n}A^{(j)}
=
A^{(n)} A^{(n-1)} \cdots A^{(1)},
\]
then, if $\|A^{(j+1)}A^{(j)}\|$ is not too small compared with $\|A^{(j+1)}\|\cdot\|A^{(j)}\|$, then the most contracted direction of $A^{(j+1)}$ is not too close to the most contracted direction of $(A^{(j)})^{-1}$. If this holds for each $1\le j<n$, then one has good norm control for the product $A$.

The Avalanche Principle allows one to move from small scales of matrix products to large scales in an inductive fashion. In particular, in Lemma~\ref{l.well-approximate-LE}, we are able to relate finite-step Lyapunov exponents at different scales using positivity of $L$ and the LDT.

\begin{lemma}[Avalanche Principle]
\label{l.avlanche-principle}
Let $A^{(1)},\ldots, A^{(n)}$ be a finite sequence in $\mathrm{SL}(2,\R)$ satisfying the following conditions:
\begin{align}\label{condition-AP}
&\min_{1\le j\le n}\|A^{(j)}\|\ge \l > n,\\ \label{condition-AP2}
&\max_{1\le j<n}\left|\log \|A^{(j+1)}\|+\log \|A^{(j)}\|-\log\|A^{(j+1)}A^{(j)}\|\right|<\frac12\log\l.
\end{align}
Then
\begin{equation} \label{avalanche-principle}
\left|\log\|A^{(n)}\ldots A^{(1)}\|+\sum_{j=2}^{n-1}\log\|A^{(j)}\|-\sum_{j=1}^{n-1}\log\|A^{(j+1)}A^{(j)}\|\right|\le C\frac{n}{\l}.
\end{equation}
\end{lemma}
\noindent See \cite[Proposition 2.2]{goldstein} for a proof of Lemma~\ref{l.avlanche-principle}.

\begin{lemma}\label{l.well-approximate-LE}
There are constants $c,C > 0$ that depend only on $\widetilde \mu$ with the property that
\begin{equation} \label{eq:approxLE}
|L(E)+L_n(E)-2L_{2n}(E)|
\leq
Ce^{-cn}
\end{equation}
for all $n \in \Z_+$ and every $E \in \hat\Sigma$.
\end{lemma}

\begin{proof}
Let $E \in \hat\Sigma$ be given; all estimates in the present argument will be uniform over $E \in \hat\Sigma$, so we will suppress it from the notation, writing $L, L_n, M_n$ in place of $L(E)$, $L_n(E)$, $M_n^E$. Throughout the argument, we let $C$ denote a large $\hat\Sigma$-dependent constant. It is straightforward to verify that the value of $C$ increases only finitely many times as the argument progresses and that it can indeed be chosen uniformly over $E \in \hat\Sigma$. Now, choose $\e > 0$ small enough that
\begin{equation} \label{eq:approxLE:epsChoice}
0
<
\frac{4\e}{\gamma-\e}
<
\frac12,
\end{equation}
where $\gamma$ is the constant from \eqref{eq:UPLE}, and pick $\eta > 0$ small enough that $\eta < 4(\gamma - \e)$ and the conclusion of Theorem~\ref{t:LDT} holds for this choice of $\eta$ and $\varepsilon$. Given $n \in \Z_+$ large, choose $\ell\in\Z_+$ so that
\begin{equation} \label{eq:approxLE:lognSimEll}
e^{\frac\eta{5}\ell}\le n\le e^{\frac\eta{4}\ell}<e^{(L-\e)\ell},
\end{equation}
where the final inequality follows from our choice of $\eta$. For each $\omega\in\Omega$, consider
\[
A^{(j)}(\omega)
:=
M_\ell\left(T^{(j-1)\ell}\omega\right),\ 1\le j \le n.
\]
From Theorem~\ref{t:LDT}, by choosing $n$, hence $\ell$, large, we obtain that there is an exceptional set $\CB = \CB(n)$ with
\[
\mu(\CB)
\leq
e^{-\frac14\eta\ell}
\]
such that
\begin{equation}\label{eq:approxLE:ellStepEst}
\left|\frac1\ell\log\|A^{(j)}(\omega)\| - L \right|,\
\left|\frac1{2\ell}\log\|A^{(j+1)}(\omega)A^{(j)}(\omega)\| - L \right|
<
\e
\end{equation}
whenever $\omega \notin \CB$ and $1 \le j \le n$. Consequently, we have
\[
\|A^{(j)}(\omega)\|\ge e^{(L-\e)\ell} > n
\]
for all $\omega \in \Omega \setminus \CB$ and $1 \le j \le n$. Moreover, for $1 \le j < n$, we get
\[
\left|\log \|A^{(j+1)}(\omega)\|+\log \|A^{(j)}(\omega)\|-\log\|A^{(j+1)}(\omega)A^{(j)}(\omega)\|\right|
<
4\e\ell
\]
from \eqref{eq:approxLE:ellStepEst}. Thus, for $\omega \in \Omega\setminus \CB$, conditions~\eqref{condition-AP} and \eqref{condition-AP2} of Lemma~\ref{l.avlanche-principle} are fulfilled upon taking $\l= \exp((L-\e)\ell)$, where we have used \eqref{eq:approxLE:epsChoice} to verify that \eqref{condition-AP2} holds true. Consequently, we obtain the conclusion of \eqref{avalanche-principle} for each $\omega\notin\CB$, which reads:
$$
\left|
\log\|M_{\ell n}(\omega)\|
+ \sum_{j=2}^{n-1} \log\| M_{\ell}(T^{(j-1)\ell}\omega) \|
- \sum_{j=1}^{n-1} \log\|M_{2\ell}(T^{(j-1)\ell}\omega) \|
\right|
\le
C\frac{n}{\l}.
$$

Dividing the inequality above by $\ell n$, integrating it against $\mu$, and splitting the domain of integration into $\Omega\setminus\CB$ and $\CB$, we get
$$
\left|L_{\ell n} + \frac{n-2}{n}L_\ell - \frac{2(n-1)}{n}L_{2\ell} \right|
\leq
C\left(\l^{-1}\ell^{-1}+\mu(\CB)\right)
\leq
Ce^{-\frac14\eta\ell}.
$$
Thus, we obtain
\begin{equation} \label{avalanche-consequence1}
\left|L_{\ell n} + L_\ell - 2L_{2\ell} \right|
\leq
\frac{C}n + Ce^{-\frac14\eta\ell}
\leq
Ce^{-\frac15\eta\ell}.
\end{equation}
Notice that \eqref{avalanche-consequence1} holds for any sufficiently large $n$ and $\ell$ that are related via \eqref{eq:approxLE:lognSimEll}. Thus, we now can build up estimates on large scales inductively. To make this precise, let us define a relation $\gg$ by declaring
\[
y \gg x
\iff
e^{\frac{\eta}{5} x}
\leq
\frac{y}{x}
\leq
\frac{1}{2} e^{\frac{\eta}{4}x}.
\]
For the first scale, we choose $n_1 \in \Z_+$ large and $n_2 \gg n_1$ a multiple of $n_1$; applying \eqref{avalanche-consequence1} with $\ell = n_1$ and $n = n_2/n_1$ yields
\begin{equation} \label{avalanche-consequence2}
\left|L_{n_2} + L_{n_1} - 2L_{2n_1}\right|
\leq
Ce^{-\frac15\eta n_{1}}.
\end{equation}
and
\begin{equation}\label{avalanche-consequence2.5}
\left|L_{2n_2} + L_{n_1} - 2L_{2n_1}\right|
\leq
Ce^{-\frac15\eta n_{1}}.
\end{equation}
Combining \eqref{avalanche-consequence2} and \eqref{avalanche-consequence2.5} gives us
\begin{equation} \label{avalanche-consequence3}
|L_{2n_2} - L_{n_2}|
\leq
Ce^{-\frac15\eta n_{1}}.
\end{equation}
Inductively, choosing $n_{s+1}\gg n_{s}$ such that $n_s | n_{s+1}$, we get \eqref{avalanche-consequence2} and \eqref{avalanche-consequence3} with the pair $(n_1,n_2)$ replaced by  $(n_s,n_{s+1})$, which in turn yields
\[
|L_{n_{s+1}} - L_{n_s}|
\leq
C e^{-\frac{1}{5} \eta n_{s-1}},
\quad \text{for every }
s \geq 2.
\]
Putting these estimates together, we obtain
\begin{align*}
0
\le
L_{n_2} - L
&\le
\sum_{s = 2}^\infty |L_{n_{s+1}} - L_{n_{s}} |\\
&\le
\sum_{s = 1}^\infty C e^{-\frac15\eta n_{s}}\\
& \leq
Ce^{-\frac15\eta n_1}.
\end{align*}
Consequently, we obtain \eqref{eq:approxLE} upon replacing $L_{n_2}$ by $L$ in \eqref{avalanche-consequence2}.
\end{proof}

Now we are ready to prove Theorem~\ref{t:HolderContinuity}.

\begin{proof}[Proof of Theorem~\ref{t:HolderContinuity}]
Let $E,E' \in \hat\Sigma$ be given. Recall that
\[
|L_n(E)-L_n(E')|
\leq
\Gamma^n|E-E'|
\]
by \eqref{eq:LnLipsch}. Combining this with \eqref{eq:approxLE}, we get
\begin{equation} \label{eq:LEHolderest1}
|L(E)-L(E')|
\leq
C^n|E-E'|+Ce^{-cn}
\end{equation}
for all large $n$, where $C,c > 0$ are suitable constants. H\"older continuity of $L$ then follows by choosing $n$ well. More precisely, with
$$
n
=
\left\lfloor \frac1{3\log C}\log\frac1{|E-E'|} \right\rfloor,
\quad
\beta
=
\min\left\{\frac{2}{3}, \frac{c}{3\log C} \right\},
$$
\eqref{eq:LEHolderest1} yields
\[
|L(E)-L(E')|
\leq
C|E-E'|^\beta,
\]
 which proves Theorem~\ref{t:HolderContinuity}.
\end{proof}

\section{Estimating Transfer Matrices and Green Functions}
\label{sec:green}

In the present section, we work out the next main thrust of the localization proof, namely, suitable upper bounds on Green functions of finite-volume truncations of $H_\omega$. In the 1D setting, the truncated Green functions are intimately connected with the transfer matrices; hence, the main technical result of the section is actually an estimate on the transfer matrices. We use the LDT to prove bounds on transfer matrices on blocks of length $n$ on a full-measure subset of $\Omega$, at the price of averaging over $n^2$ consecutive blocks. In fact, we prove a ``centered'' version of this result that allows us to shift the result to the center of localization (once we know that the eigenfunctions are localized, that is). We then parlay this result into an upper bound on the Green functions. Recall the function $F_n$ defined in \eqref{eq:FnOmegaDef}, that is,
\[
F_n(\omega,E)
=
\frac{1}{|n|} \log\|M_n^E(\omega)\|.
\]

\begin{lemma}\label{lem:mean_ldt}
For any $\e > 0$, there exists $n_0 = n_0(\e)$ such that
\begin{equation}\label{eq:mean_ldt}
\mu\set{\omega: \vabs{L(E) - \frac{1}{r}\sum\limits_{s = 0}^{r-1} F_n(T^{sn+\ell} \omega, E)} \geq \e}
\leq
e^{-\frac{\e r}{2}}
\end{equation}
for every $\ell \in\Z$, $E\in \hat\Sigma$, $r\in \Z_+$, and every $n \geq n_0$.
\end{lemma}
\begin{proof}
Since $\mu$ is $T$-invariant, we only need to deal with $\ell = 0$. Fix $E \in\hat\Sigma$ and $\e > 0$, and suppress $E$ from the notation. By the triangle inequality and Chebyshev's inequality, we have
\begin{align*}
\mu\set{\omega: \vabs{L - \frac{1}{r}\sum\limits_{s = 0}^{r-1} F_n(T^{sn} \omega)} \geq \e}
& \leq
\mu\set{\omega: \frac{1}{r}\sum\limits_{s = 0}^{r-1} \vabs{L - F_n(T^{sn} \omega)} \geq \e} \\
& \leq
e^{-r\e}
\int_\Omega \! \exp\left( \sum_{s=0}^{r-1} \left| L - F_n \circ T^{sn} \right| \right) \, \dd\mu.
\end{align*}
Since $F_n\circ T^{sn}$ and $F_n \circ T^{s'n}$ depend on disjoint sets of coordinates whenever $s \neq s'$, it follows $e^{|L-F_n\circ T^{sn}|}$ are indpendent random variables on $\Omega$ for distinct $s$. Thus,
\[
\int_\Omega \! \exp\left( \sum_{s=0}^{r-1} \left| L - F_n \circ T^{sn} \right| \right) \, \dd\mu=\prod_{s=0}^{r-1} \int_\Omega \! e^{|L-F_n\circ T^{sn}|} \, \dd\mu
\]
which in turn implies
\begin{align}
\nonumber
\mu\set{\omega: \vabs{L - \frac{1}{r}\sum\limits_{s = 0}^{r-1} F_n(T^{sn} \omega)} \geq \e}
& \leq
e^{-r\e} \prod_{s=0}^{r-1} \int_\Omega \! e^{|L-F_n\circ T^{sn}|} \, \dd\mu
 \\
 \label{eq:sample_mean_deviation}
& =
\left(e^{-\e}\int_{\Omega} \! e^{\vabs{L - F_n}} \, \dd\mu \right)^{r},
\end{align}
where the second line follows from $T$-invariance of $\mu$.

It remains to bound the integral on the right-hand side of \eqref{eq:sample_mean_deviation}. To that end, take $\d < \e/2$, split $\Omega$ into the regions where $|L - F_n| \ge \d$ and $|L - F_n| < \d$, and then apply Theorem~\ref{t:LDT} to estimate the measure of the former region:
\begin{equation} \label{eq:meanLDT:integralEst}
\int_{\Omega} \! e^{\vabs{L - F_n(\omega)}} \, \dd\mu(\omega)
\leq
C \Gamma^2 e^{-\eta n} + e^\d.
\end{equation}
Since $\d < \e/2$, we may choose $n_0 = n_0(\e)$ large enough that
\begin{equation} \label{eq:meanLDT:n0Choice}
C \Gamma^2 e^{-\eta n_0} + e^\d
<
e^{\frac{\e}{2}}.
\end{equation}
In view of \eqref{eq:sample_mean_deviation}, \eqref{eq:meanLDT:integralEst}, and \eqref{eq:meanLDT:n0Choice}, we have
\begin{equation*}
\mu\set{\omega: \vabs{L - \frac{1}{r}\sum\limits_{s = 0}^{r-1} F_n(T^{sn} \omega)} \geq \e} \leq e^{-\frac{\e r}{2}}
\end{equation*}
for all $n \geq n_0$ and all $r \in \Z_+$.
\end{proof}

Using the lemma, we can get a full-measure set of $\omega \in \Omega$ on which one can control transfer matrices across blocks of length $n$, if one is willing to average over $n^2$ consecutive blocks.

\begin{prop} \label{prop:meanLDTshifted_sigma}
For any $0 < \e < 1$, there exists a subset $\Omega_+ = \Omega_+(\e) \subset \Omega$ of full $\mu$-measure such that the following statement holds true. For every $\e \in (0,1)$ and every $\omega \in \Omega_+(\e)$, there exists $\widetilde n_0 = \widetilde n_0(\omega,\e)$ such that
\begin{equation}\label{eq:partsum1}
\left|L(E) - \frac{1}{n^2}\sum_{s=0}^{n^2-1} F_{n}\left(T^{\zeta_0 + sn}(\omega),E\right) \right|
<
\e
\end{equation}
for all $n,\zeta_0 \in \Z$ with $n \geq \max\left( \widetilde n_0, (\log(|\zeta_0|+1))^{2/3} \right)$ and all $E \in \hat \Sigma$.
\end{prop}
\begin{proof}
The key realization in this proof is that one can obtain good control of the averages of the functions $F_n(\omega,E)$ globally over $E \in \hat\Sigma$ by controlling $F_n$ on suitable finite subsets of $\hat\Sigma$ whose cardinality can in turn be bounded via the perturbative estimates from Lemma~\ref{lem:Fnlipschitz} and compactness of $\hat\Sigma$.

For a given large enough $n$ (we will determine largeness later), we first consider sets $\CB_{n,\zeta_0} = \CB_{n,\zeta_0}(\varepsilon)$ where \eqref{eq:partsum1} fails to hold:
\begin{align*}
\CB_{n,\zeta_0}
:=
\left\{\omega : \sup_{E \in \hat\Sigma}
\left|L(E) - \frac{1}{n^2}\sum_{s=0}^{n^2-1} F_n(T^{\zeta_0 + sn}\omega,E) \right|
\geq
\e \right\}.
\end{align*}

Given $0 < \delta \leq 1/2$, define the na\"ive grid
\[
\Sigma_0
:=
\left[ \hat\Sigma \cap \left( 2 \delta \, \Z \right) \right] \cup \set{\pm \kappa}.
\]
It is straightforward to check that $\Sigma_0$ is $\delta$-dense in $\hat\Sigma$ in the sense that
\begin{equation} \label{eq:prop82:deltaDense_centered}
\hat\Sigma
\subseteq
\bigcup_{t \in \Sigma_0} [t-\delta,t+\delta].
\end{equation}
Moreover, we may estimate the cardinality of $\Sigma_0$ via
\[
\# \Sigma_0
\leq
\frac{\kappa}{\delta} + 3
\leq
\frac{2\kappa}{\delta};
\]
note that we used $\delta \leq 1/2$ and $\kappa \ge 2$ in the second step. Now, fix $0 < \e < 1$ and let $\Gamma$ denote the uniform bound on $\|M\|$ from \eqref{eq:oneStepTMBoundDef}. Taking $\delta = \e (3 \Gamma^{n})^{-1}$ in the discussion above, we may produce a finite set $\Sigma_0 \subset \hat{\Sigma}$ which is $\e (3 \Gamma^{n})^{-1}$-dense in $\hat\Sigma$ in the sense of \eqref{eq:prop82:deltaDense_centered}, with cardinality bounded above by
\begin{equation} \label{eq:cardinality_centered}
\#\Sigma_0
\leq
\frac{6\kappa \Gamma^{n}}{\e}.
\end{equation}
If necessary, enlarge $n$ to ensure that
\begin{equation}\label{eq:holderApplication}
C \left(\frac{\e}{3\Gamma^n} \right)^\beta
<
\frac{\e}{3},
\end{equation}
where $C$ and $\beta$ are from Theorem~\ref{t:HolderContinuity}.
Then \eqref{eq:Fnderiv} and Theorem~\ref{t:HolderContinuity} yield
\begin{align*}
\CB_{n,\zeta_0}
\subset
\bigcup_{E \in \Sigma_0}
\left\{\omega: \left|L(E) - \frac{1}{n^2}\sum_{s=0}^{n^2-1} F_{n} \left(T^{\zeta_0 + sn}\omega,E \right) \right|
\geq
\frac{\e}{3} \right\}.
\end{align*}
Consequently, by taking $n$ large enough that $n \geq n_0(\e/3)$ and \eqref{eq:holderApplication} holds and using $T$-invariance of $\mu$, one obtains
\begin{equation} \label{eq:supersmall}
\mu(\CB_{n,\zeta_0})
\leq
\frac{6 \kappa \Gamma^{n}}{\e} \exp\left(-\frac{\e n^{2}}{6} \right)
\end{equation}
by Lemma~\ref{lem:mean_ldt}. Now, with
\[
\CB_{n}
:=
\bigcup\limits_{|\zeta_0| \leq e^{n^{3/2}}} \CB_{n, \zeta_0},
\]
it is clear that $\mu\left(\CB_{n}\right) \leq e^{-c\e n^2}$ for $n$ large, so $\Omega_+ := \Omega \setminus \limsup \CB_{n}$ satisfies $\mu(\Omega_+) = 1$ by the Borel--Cantelli Lemma. Naturally, for each $\omega \in \Omega_+$, we can find $\widetilde n_0 = \widetilde{n}_0(\omega, \e)$ large enough that $\omega \notin \CB_n$ whenever $n \ge \widetilde n_0$. In other words,
\[
\omega \notin \CB_{n, \zeta_0} \mbox{ whenever } n\geq \widetilde{n}_0(\omega, \e) \mbox{ and } |\zeta_0| \leq e^{n^{3/2}}.
\]
Changing the order of $n$ and $\zeta_0$, the statement above clearly implies
\[
\omega \notin \CB_{n,\zeta_0}
\mbox{ whenever } \zeta_0 \in\Z \mbox{ and } n \geq \max\left( \widetilde n_0, (\log(|\zeta_0|+1))^{2/3} \right).
\]
By the definition of $B_{n,\zeta_0}$, we obtain the statement of the proposition.
\end{proof}

We are now in a position to estimate the finite-volume Green functions. Before stating the estimate, we fix some notation. Let $\Lambda = [a,b] \cap \Z$ be a finite subinterval of $\Z$, and denote by $P_\Lambda:\ell^2(\Z) \to \ell^2(\Lambda)$ the canonical projection. We denote the restriction of $H_\omega$ to $\Lambda$ by
$$
H_{\omega,\Lambda}
:=
P_\Lambda H_\omega P_\Lambda^*.
$$
For any $E \notin \sigma(H_{\omega,\Lambda})$, define
$$
G_{\omega,\Lambda}^E
:=
(H_{\omega,\Lambda}-E)^{-1},
$$
to be the resolvent operator associated to $H_{\omega,\Lambda}$.  Like $H_{\omega,\Lambda}$, $G_{\omega,\Lambda}^E$ has a representation as a finite matrix; denote its matrix elements by $G_{\omega,\Lambda}^E(m,n)$, that is,
$$
G_{\omega,\Lambda}^E(m,n)
:=
\left\langle \delta_m , G_{\omega,\Lambda}^E \delta_n \right\rangle,
\quad
n,m \in \Lambda.
$$
Additionally, for $N \in \Z_+$, let us define $H_{\omega,N} := H_{\omega,[0,N)}$ to be the restriction of $H_\omega$ to the box $\Lambda_N := [0,N)\cap \Z.$  We will likewise use the same notation for the associated resolvent $G_{\omega,N}^E$.
\medskip

Using Cramer's rule, we know that
\begin{equation} \label{eq:greenstodet}
G_{\omega,N}^E(j,k)
=
\frac{\det[H_{\omega,j} - E] \det[H_{T^{k+1}\omega, N-k-1} - E]}{\det[H_{\omega,N} - E]}
\end{equation}
 for any $0 \leq j \leq k \leq N-1$ and $E \notin \sigma(H_{\omega,N})$, where one interprets $\det[H_{\omega,0} - E] = 1$.

Another relation that will be important in what follows is
\begin{align}
\label{eq:transfertodet}
M_N^E(\omega)
=
\begin{bmatrix}
\det(E - H_{\omega,N}) & -\det(E - H_{T\omega,N-1})\\
\det(E - H_{\omega,N-1}) & -\det(E - H_{T\omega,N-2})
\end{bmatrix},
\quad
N \ge 2.
\end{align}
This is a standard fact, which one may prove inductively.

In particular, since the norm of a matrix majorizes the absolute value of any of its entries, we obtain
\begin{equation} \label{eq:greenBoundsFromTMBounds}
\left| G_{\omega,N}^E(j,k) \right|
\leq
\frac{\|M_{j}^E(\omega)\| \|M_{N-k}^E(T^{k}\omega)\|}{|\det[H_{\omega,N} - E]|}
\end{equation}
for all $0 \leq j \leq k \leq N-1$ by combining \eqref{eq:greenstodet} and \eqref{eq:transfertodet}. Thus, it is straightforward to transform estimates on transfer matrices into estimates on Green functions of truncations of $H_\omega$; to complete our goal of estimating Green functions, we will use Proposition~\ref{prop:meanLDTshifted_sigma} to estimate transfer matrix norms, and then apply \eqref{eq:greenBoundsFromTMBounds}.

\begin{corollary}\label{cor:LEandGreenEst_general}
Given $\varepsilon \in (0,1)$ and $\omega \in\Omega_+(\varepsilon)$, there exists $\widetilde n_1 = \widetilde n_1(\omega,\varepsilon)$ so that the following statements hold true.  For all $E \in \hat\Sigma$, we have
\begin{equation}\label{eq:FiniteLyapUpperEstimate_general}
\frac{1}{n}\log
\norm{M_n^E(T^{\zeta_0}\omega)}
\leq
L(E) + 2\varepsilon
\end{equation}
whenever $n, \zeta_0 \in \Z$ satisfy $n \geq \max\left( \widetilde n_1, \log^2(|\zeta_0| + 1) \right)$.

Moreover, for all  $n, \zeta_0 \in \Z$ with $n\ge \e^{-1}\max\left( \widetilde n_1, 2\log^2(|\zeta_0|+1)\right)$, we have
\begin{equation}\label{eq:greenEst_K_general}
\left| G_{T^{\zeta_0} \omega,n}^E(j,k) \right|
\leq
\frac{\exp[(n - |j - k|)L(E) + C_0 \e n]}{|\det[H_{T^{\zeta_0} \omega,n} - E]|}
\end{equation}
for all $E \in \hat \Sigma \setminus \sigma(H_{T^{\zeta_0} \omega,n})$ and all $j,k\in [0,n)$, where $C_0$ is a constant that only depends on $\widetilde\mu$.
\end{corollary}

\begin{proof}
Fix $\e \in (0,1)$, $\omega \in \Omega_+(\e)$, and $E \in \hat\Sigma$. As usual, our estimates are independent of the energy, so we suppress $E$ from the notation. Choose $\widetilde n_1 \in \Z_+$ large enough that
\begin{equation}
\label{eq:AveragingCorollary_general_N_Largeness}
\widetilde n_1
\geq
\max\left( \widetilde n_0(\omega,\e)^3, 4\e^{-1}, 2986 \right), \quad
\text{and} \quad
\frac{12 \log\Gamma}{\widetilde n_1^{1/3} - 3} < \e.
\end{equation}
Given $n,\zeta_0 \in \Z$ with $n \geq \max\left(\widetilde n_1 ,\log^2(|\zeta_0|+ 1) \right)$, we want to apply Proposition~\ref{prop:meanLDTshifted_sigma} with $m = \lceil n^{1/3}\rceil$. Notice that $0 \le m^3 -n \leq 3m^2$. Thus, by submultiplicativity of the matrix norm and unimodularity of the transfer matrices, we have
\begin{equation}\label{eq:PartialLyapCor}
\|M_{n}(T^{\zeta_0}\omega)\|
\leq
\Gamma^{3m^2}\prod_{s = 0}^{m^2-1} \|M_m(T^{\zeta_0 + sm} \omega)\|.
\end{equation}
By our choice of $\widetilde n_1$, we have $m \geq \widetilde n_0$ and $m \geq (\log(|\zeta_0|+1))^{2/3}$. Thus, combining \eqref{eq:PartialLyapCor} with Proposition~\ref{prop:meanLDTshifted_sigma}  and \eqref{eq:AveragingCorollary_general_N_Largeness}, a direct computation shows that
\begin{equation}\label{eq:lognorm_estimate_general}
\begin{split}
\frac{1}{n}\log\|M_{n}(T^{\zeta_0} \omega)\|
& \leq
\frac{3m^2 \log \Gamma}{n} + \frac{m^3}{n}(L+\e) \\
& \leq
L + \frac{6\log\Gamma}{m-3} + \frac{m}{m-3} \e \\
& \leq
L + 2\e,
\end{split}
\end{equation}
which yields \eqref{eq:FiniteLyapUpperEstimate_general}.

Now suppose $n\ge \e^{-1}\max \left( \widetilde n_1, 2 \log^2(|\zeta_0|+1)) \right)$ and put $h := \lceil \varepsilon n\rceil$. For $j \ge 0$, we have
\begin{equation*}
\norm{ M_j(T^{\zeta_0} \omega) }
\leq
\norm{ M_{j + h}( T^{\zeta_0 - h} \omega ) }
\norm{ [M_h( T^{\zeta_0 - h} \omega )]^{-1} }.
\end{equation*}
Clearly $j \ge 0$ and $h \ge \widetilde n_1$. Moreover, by our choice of $\widetilde n_1$ and the relation between $n$ and $\zeta_0$, a direct computation shows $h \geq \log^2(|\zeta_0 - h| + 1)$. 
 Thus we can apply \eqref{eq:lognorm_estimate_general} to estimate the norms on the right hand side and obtain
\begin{equation}\label{eq:TransferMatrixFromGreenF1}
\norm{ M_j(T^{\zeta_0} \omega) }
\leq
e^{(j + 2h)L(E) + 2\e(j + 2h)}
\leq
e^{jL(E) + C_0\varepsilon n},
\end{equation}
where $C_0$ is a constant that depends only on $\widetilde \mu$.

Naturally, one can also apply the analysis above to estimate the transfer matrix $M_{n-k}(T^{\zeta_0 + k} \omega)$ with $ j\le k \le n-1$ as well. Using \eqref{eq:AveragingCorollary_general_N_Largeness} and the relationships among $h$, $\zeta_0$ and $k$, a direct computation shows that $\log^2(|\zeta_0 - k + h| + 1)\le h$. Thus, \eqref{eq:lognorm_estimate_general} yields
\begin{align}
\nonumber
\| M_{n-k}(T^{\zeta_0 + k} \omega) \|
& \leq
\left\| M_{n-k+h}(T^{\zeta_0 + k - h} \omega) \right\|
\left\| \left[ M_{h}(T^{\zeta_0 + k - h} \omega)\right]^{-1}\right \| \\
\label{eq:TransferMatrixFromGreenF2}
& \leq
\exp\left[(n-k)L(E) + C_0 \e n\right].
\end{align}
Combining equations \eqref{eq:TransferMatrixFromGreenF1} and \eqref{eq:TransferMatrixFromGreenF2} with the observation \eqref{eq:greenBoundsFromTMBounds}, one sees that for a suitable choice of $C_0$
\begin{align*}
|G_{T^{\zeta_0} \omega,n}^E(j,k)|
& \leq
\frac{\|M_{j}(T^{\zeta_0}\omega)\| \|M_{n-k}(T^{\zeta_0 + k} \omega)\|}{|\det[H_{T^{\zeta_0} \omega,n} - E)]|} \\
& \leq
\frac{\exp[(n - |j - k|)L(E) + C_0 \e n]}{|\det[H_{T^{\zeta_0} \omega,n} - E)]|},
\end{align*}
for all $E \in \hat \Sigma \setminus \sigma(H_{T^{\zeta_0} \omega,n})$ and all $0 \le j \le k < n$. The case $j \geq k$ follows because $H$ is self-adjoint and $E$ is real.
\end{proof}

\section{Proof of Spectral and Exponential Dynamical Localization}\label{sec:localization}

In the present section, we conclude the proof of Theorem~\ref{t:LyapConvGenEvals}, which, as discussed in the Introduction, implies Theorem~\ref{t.main} by standard reasoning. We also state and prove Theorem~\ref{t:DL}, which contains our exponential dynamical localization result.

The key remaining cornerstone is supplied by Proposition~\ref{prop:doubleRes}, which is a version of an argument usually referred to as the \emph{elimination of double resonances}. We note that double resonances appear frequently in, and are one of the most subtle parts of, the mathematical analysis of Anderson localization. In particular, from the proof of the present paper, or the proof of \cite{bg,bs}, one may see that, for one-dimensional ergodic Schr\"odinger operators, uniform positivity and uniform LDT of the Lyapunov exponent, together with the elimination of double resonances, imply Anderson localization for suitable parameters. Moreover, the set of parameters (phases, frequencies) for which  Anderson localization holds true depends exactly on the bad set that is eliminated by the exclusion of double resonances.

It is also interesting to note that, in some sense, the elimination of double resonances detects the randomness of the base dynamics in a very sensitive manner. For instance, a version of Proposition~\ref{prop:doubleRes} was developed by Bougain-Goldstein for real-analytic quasiperiodic potentials in \cite{bg}, where the authors invoked suitable complexity bounds for \emph{semi-algebraic} sets. Bourgain-Schlag \cite{bs} considered operators with strongly mixing potentials, where the elimination becomes a bit easier due to the strong mixing property. Obviously the Anderson model considered in the present paper is closer to the one in \cite{bs}, and the elimination process is even more transparent because of the independence of the potential values.

After we prove Proposition~\ref{prop:doubleRes}, this result is then used to supply estimates that enable us to run the Avalanche Principle and prove positivity (and existence) of \emph{non-averaged} Lyapunov exponents
\[
L(E,\omega)
:=
\lim_{n\to\infty} \frac{1}{n} \log \|M_n^E(\omega)\|.
\]
Once we have existence and positivity of $L(E,\omega)$ for a full-measure set of $\omega\in\Omega$ and a suitably rich ($\omega$-dependent!) set of energies, we are able to deduce the modified Anderson localizer's dream. Thus, the Avalanche Principle in some sense plays the role of MSA in our arguments. As discussed in the Introduction, the Anderson localizer's dream itself is false, that is, one does not have positivity (or even existence!) of $L(E,\omega)$ for all $E$ and a uniform full-measure set of $\omega \in \Omega$. After proving the modified Anderson localizer's dream, we can then make a second pass through the localization argument to obtain better estimates; in particular, we obtain ``centered'' versions of the localization estimates with constants that depend only on $\omega$ and the center of localization. This supplies a sufficient input to deduce almost-sure exponential dynamical localization.

For $N \in \Z_+$, we define
\[
\overline{N}
:=
\left\lfloor N^{\log N}\right\rfloor
=
\left\lfloor e^{(\log N)^2} \right\rfloor,
\]
which is a super-polynomially and subexponentially growing function of $N$.
\medskip

We now introduce the set of double resonances. Given $\e > 0$ and $N \in \Z_+$, let $\CD_N = \CD_N(\e)$ denote the set of all those $\omega \in \Omega$ such that
\begin{equation} \label{eq:doubleRes:GreenLB_centered}
\|G_{T^\zeta\omega,[-N_1, N_2]}^E\|
\geq
e^{K^2}
\end{equation}
and
\begin{equation} \label{eq:doubleRes:FNupperBound_centered}
|F_m(T^{\zeta + r}\omega,E)|
\leq
L(E)-\e.
\end{equation}
for some choice of $\zeta \in \Z$, $K \ge \max(N, \log^2(|\zeta|+1))$, $0\le N_1,N_2 \leq K^9$, $E \in \Sigma$, $K^{10} \le r \leq \overline K$, and $m \in \{K,2K\}$ ($F_m$ is as defined in \eqref{eq:FnOmegaDef}).

\begin{prop} \label{prop:doubleRes}
For all $0 < \e < 1$, there exist constants $C > 0$ and $\widetilde\eta > 0$ such that
\[
\mu(\CD_N(\e))
\leq
Ce^{-\widetilde\eta N}
\]
for all $N \in \Z_+$.
\end{prop}

\begin{proof}
Define auxiliary ``bad sets'' for fixed $\zeta$ and $K$:
\begin{equation*}
\CD_{K,\zeta} = \set{\omega: \eqref{eq:doubleRes:GreenLB_centered}, \eqref{eq:doubleRes:FNupperBound_centered}\text{ are satisfied for some choice of } E,N_1,N_2,r,m \text{ as above}}.
\end{equation*}
Fix $\e \in (0,1)$ and begin by noticing that
\begin{equation}\label{eq:split-into-N-k-sets_centered}
\CD_{K,\zeta}
\subset
\bigcup_{K^{10} \le r \le \overline{K}}\;\;
\bigcup_{0 \leq N_1,N_2\leq K^{9}}
\widetilde\CD_1(N_1,N_2,r,\zeta) \cup \widetilde\CD_2(N_1,N_2,r,\zeta),
\end{equation}
where $\widetilde\CD_j(N_1,N_2,r,\zeta)$ denotes the collection of all $\omega \in \Omega$ for which there exists $E \in \Sigma$ such that \eqref{eq:doubleRes:GreenLB_centered} and \eqref{eq:doubleRes:FNupperBound_centered} hold with $m=jK$. We will estimate $\mu(\widetilde\CD_1)$. The estimates for $\widetilde\CD_2$ are completely analogous. To that end, suppose $\omega\in\widetilde\CD_1(N_1,N_2,r,\zeta)$, i.e.\ \eqref{eq:doubleRes:GreenLB_centered} and \eqref{eq:doubleRes:FNupperBound_centered} hold for some $E \in \Sigma$. By the spectral theorem, there exists $E_0 \in \sigma(H_{T^\zeta\omega,[-N_1,N_2]})$ with
\begin{equation}
\label{eq:doubleRes:DtildeUB}
|E - E_0|
\leq
\left\|G^E_{T^\zeta\omega,[-N_1,N_2]} \right\|^{-1}
\leq
e^{-K^{2}}.
\end{equation}
On the other hand, choosing $K$ large enough that $\Gamma^{K} e^{-K^{2}}\leq \frac{\e}{4}$ and $ Ce^{-\beta K^{2}} \leq \frac{\e}{4}$ (where $C,\beta$ are from \eqref{eq:holder}), we get
\begin{align*}
F_{K}(T^{\zeta+r}\omega,E_0)
& \leq
F_{K}(T^{\zeta+r}\omega, E) + \frac{\e}{4}\\
& \leq
L(E) - \e + \frac{\e}{4}\\
& \leq
L(E_0) - \frac{\e}{2},
\end{align*}
where we have used Lemma~\ref{lem:Fnlipschitz} in the first line, \eqref{eq:doubleRes:FNupperBound_centered} in the second line, and Theorem~\ref{t:HolderContinuity} in the final line.
Thus, when $K$ is large enough, we get
\begin{equation} \label{eq:doubleRes:DtildeUB_centered}
\widetilde\CD_1(N_1,N_2,r,\zeta)
\subseteq
\hat \CD_1(N_1,N_2,r,\zeta) \text{ for all } N_1, N_2, r, \text{ and } \zeta,
\end{equation}
where $\hat \CD_1 = \hat \CD_1(N_1,N_2,r,\zeta)$ denotes the set of all $\omega \in \Omega$ such that
\[
F_{K}(T^{\zeta+r}\omega,E_0)
\leq
L(E_0) - \frac{\e}{2}
\]
for some $E_0\in\sigma\left(H_{T^\zeta\omega,[-N_1,N_2]} \right)$.

Now, the conditions \eqref{eq:doubleRes:GreenLB_centered} and \eqref{eq:doubleRes:FNupperBound_centered} depend only on a finite number of entries of $\omega$. Concretely, we notice that $G_{T^\zeta\omega,[-N_1,N_2]}^E$ depends on $\vec\omega' := (\omega_{\zeta-N_1},\ldots,\omega_{\zeta + N_2})$, while $M_{K}^E(T^{\zeta+r} \omega)$ depends only on $\vec\omega'' := (\omega_{\zeta+r},\ldots,\omega_{\zeta + r+K-1})$. In particular, \eqref{eq:doubleRes:GreenLB_centered} and \eqref{eq:doubleRes:FNupperBound_centered} depend on independent sets of random variables. Consequently, we obtain
\begin{align*}
\int_{\CA^{[\zeta+r,\zeta+r+\base)}} \! \chi_{\hat\CD_1}(\omega) \, \dd\widetilde\mu^{\base}(\vec\omega'')
\leq
C(N_1+N_2+1) e^{-\eta_1 \base}
\end{align*}
for each fixed choice of $\vec\omega' \in \CA^{[-N_1 + \zeta, N_2 + \zeta]}$. Then, since membership in $\hat \CD$ is determined entirely by coordinates in $\zeta+[-N_1,N_2]$ and $\zeta+[r,r+\base)$, we have
\begin{align*}
\mu(\hat\CD_1)
& =
\int_\Omega \! \chi_{\hat\CD_1}(\omega) \, \dd\mu(\omega) \\
& =
\int\limits_{\CA^{\zeta+[-N_1,N_2]}} \int\limits_{\CA^{\zeta+[r,r+\base)}} \! \chi_{\hat\CD}(\omega) \, \dd\widetilde\mu^{\base}(\vec\omega'') \, \dd\widetilde\mu^{N_1+N_2+1}(\vec\omega') \\
& \leq
C \base^{9} e^{-\eta_1 \base} \\
& \leq
Ce^{- \eta_2 \base}.
\end{align*}
Thus, we obtain $\mu(\widetilde\CD_1(N_1,N_2,r,\zeta)) \leq Ce^{-\eta_2 \base}$ by applying \eqref{eq:doubleRes:DtildeUB_centered}. Applying similar reasoning to $\widetilde\CD_2$, one can estimate $\mu(\widetilde \CD_2(N_1,N_2,r,\zeta) ) \leq Ce^{-\eta_3\base}$; putting everything together yields:
\begin{align*}
\mu(\CD_{K,\zeta})
& \leq
\sum_{0 \leq N_1,N_2 \leq \base^{9}} \;
\sum_{K^{10} \le r \le \overline{K}}
\left(\mu(\widetilde\CD_1(N_1,N_2,r)) + \mu(\widetilde\CD_2(N_1,N_2,r)) \right) \\
& \leq
C \base^{18} \overline{\base} e^{-\eta_4 \base} \\
& \leq
Ce^{-2\widetilde\eta \base},
\end{align*}
for some suitable choice of $\widetilde\eta$. Changing the order of $K$ and $\zeta$, we have
\[
\CD_N
=
\bigcup_{\zeta\in\Z}\ \bigcup_{K\ge\max\{N,\log^2(|\zeta|+1)\}}\CD_{K,\zeta}
\subseteq
\bigcup_{K \ge N} \bigcup_{|\zeta| \le e^{\sqrt K}} \CD_{K,\zeta}.
\]
Then, the estimates above yield
\[
\mu(\CD_N)
\leq
\sum_{K \ge N} (2e^{\sqrt K}+1) Ce^{-2\widetilde\eta K}
\leq
Ce^{-\widetilde\eta N}
\]
for large enough $N$. Adjusting the constants to account for small $N$ concludes the proof.
\end{proof}

By Proposition~\ref{prop:doubleRes}, the set
\[
\Omega_{-}
=
\Omega_{-}(\e)
:=
\Omega \setminus \limsup \CD_N(\e),
\]
has full $\mu$-measure. We are now in a position to define the full-measure set upon which Anderson localization holds. First, put
\[
\Omega_0
=
\Omega_\Sigma \cap
\bigcap_{\e \in (0,1)} \Omega_+(\e) \cap \Omega_-(\e),
\]
where $\Omega_+(\e)$ is as defined in Proposition~\ref{prop:meanLDTshifted_sigma} and $\Omega_\Sigma$ denotes the set of $\omega \in \Omega$ for which $\sigma(H_\omega) = \Sigma$. In essence, $\omega \in \Omega_+$ gives us upper bounds on $F_N(\omega,E)$ on $\hat\Sigma$, while $\omega \in \Omega_-$ will give us lower bounds on $F_N$, provided we can prove an estimate like \eqref{eq:doubleRes:GreenLB_centered}. Then, we define
\[
\Omega_*
=
\Omega_0 \cap \CR[\Omega_0],
\]
where $\CR$ denotes the reflection $[\CR\omega]_n = \omega_{-1-n}$.
It is straightforward to verify that $\Omega_*$ has full $\mu$-measure.
\smallskip

We will need the following standard formula that relates solutions of the difference equation \eqref{eq:deve} and truncated Green functions at energy $E$. Suppose $a \leq b$ are integers, and $n\in \Lambda := [a,b] \cap \Z$. If $u$ is a solution of the difference equation $Hu = Eu$ in the sense of \eqref{eq:deve} and $E \notin\sigma(H_{\Lambda})$, then
\begin{equation}\label{eq:deveGreen}
u(n)
=
-G^E_{\Lambda}(n,a)u(a-1) - G^E_{\Lambda}(n,b)u(b+1).
\end{equation}

We now have all of the necessary pieces to do the following two things:

\begin{enumerate}
\item  Prove the modified Anderson localizer's dream. That is, we will show that generalized eigenfunctions exhibit Lyapunov behavior for almost every $\omega \in \Omega$ and every $E\in \CG(H_\omega)$.

\item The following version of SULE (semi-uniformly localized eigenfunctions), from which we will deduce exponential dynamical localization for $\mu$-a.e.\ $\omega$.
\end{enumerate}

\begin{theorem}[SULE]\label{t:sule}
For every $\delta > 0$ and $\omega \in \Omega_*$, there exist constants $C_\delta$ and $C_{\omega,\delta}$ such that
 for every eigenfunction $u$ of $H_\omega$ and every $n \in \Z$
\begin{equation}\label{sule}
|u(\bar \zeta+n)|
\leq
C_{\omega,\delta}
\|u\|_\infty
e^{C_\delta \log^{22}(|\bar \zeta|+1)}e^{-(1-\delta)L(E)|n|}
	\end{equation}
	for some $\bar \zeta$ that depends on $u$.
\end{theorem}

It turns out that these two can be done via two passes through the same argument. We begin with a generalized eigenfunction $u$ of $H_\omega$ at energy $E\in \CG(H_\omega)$, obeying \eqref{eq:polyBounds}. Since $u$ cannot vanish identically, we can pick $\zeta\in\Z$ such that $u(\zeta) \neq 0$ and normalize $u$ so that $u(\zeta) = 1$. In fact, at this stage, we may choose $\zeta=0$ or $1$. Since $\omega \in\Omega_0$, we have $\omega \in \widetilde\Omega^{(N)} := \Omega_+(\e) \setminus\CD_N(\e)$ for all $N\ge N_0$, where $N_0$ is sufficiently large.

We will then use Proposition~\ref{prop:doubleRes}, the bounds on the norms of the Green functions, and the Avalanche Principle in a somewhat subtle fashion. Initially, since we are dealing with $\omega$ in the good set $\widetilde\Omega^{(N)}$, our goal is to establish that there are some $0\le N_1, N_2\le K^{9}$ such that \eqref{eq:doubleRes:GreenLB_centered} holds; that is
\begin{equation}\label{eq:main:GreenLB}
\|G_{T^\zeta \omega,[-N_1, N_2]}^E \|
\geq
e^{K^{2}}.
\end{equation}
Having established \eqref{eq:main:GreenLB}, it then follows that expression \eqref{eq:doubleRes:FNupperBound_centered} fails for all choices of $K^{10} \leq r \leq \bar K$ and $m \in \{K,2K\}$. This supplies lower bounds on $F_K$ and $F_{2K}$ that we can then use to run the Avalanche Principle. This will lead to the modified Anderson localizer's dream and hence exponential decay of $u$. In this process, the constants and largeness conditions depend not only on $\omega$ and $\e$, but also on $u$. However, once we know that $u$ is exponentially decaying, we can pass it through the same argument centered around its global maximum to get much better control on the constants involved; in particular, we get uniformity of constants in $u$ at the expense of introducing a constant that grows subexponentially in the center of localization.

We wish to emphasize that to get spectral localization, it suffices to deal with $\zeta = 0$ or $1$, from which we may get estimates like $|u(\zeta+n)|\leq C(\omega,u)e^{-\xi n}$; here the constants depend not only on $\omega$ but also on the eigenfunction $u$. Thus we cannot get any meaningful estimates of dynamical quantities from the ``first pass'' through the localization argument, since such quantities tend to involve all eigenfunctions at once.

To tackle this issue, we need the full strength of the statements from Section~\ref{sec:green} and the present section to get all $\zeta\in\Z$ involved. In particular, we will see that the base scale for a function localized at $\zeta$ is $K\ge \e^{-1}\max(N, 2 \log^2(|\zeta|+1))$ for a suitable choice of $\e > 0$ and some large $N$ independent of $\zeta$.

\begin{proof}[Proof of Theorem~\ref{t:LyapConvGenEvals} {\rm(}modified Anderson localizer's dream{\rm)}]
We will show that $\Omega_*$ is the desired full-measure subset of $\Omega$. So, let $\omega \in \Omega_*$ and $E \in \CG(H_\omega)$ be given. Since $\Omega_*$ is $\CR$-invariant and $\CG(H_{\CR \omega}) = \CG(H_\omega)$, it suffices to show that
\begin{equation} \label{eq:main:LyapConvRHL}
\lim_{N\to\infty} \frac{1}{N}\log\|M_N^E(\omega)\|
=
L(E).
\end{equation}
Since $\omega \in \Omega_+(\e)$ for each $\e>0$, Corollary~\ref{cor:LEandGreenEst_general} yields
\[
\limsup_{N\to\infty} \frac{1}{N}\log\|M_N^E(\omega)\|
\leq
L(E),
\]
so, to prove \eqref{eq:main:LyapConvRHL}, it remains to show that
\begin{equation} \label{eq:main:FNLB}
\liminf_{N\to\infty} \frac{1}{N}\log\|M_N^E(\omega)\|
\geq
L(E).
\end{equation}
To that end, let $\e \in (0,1)$ be given and let $u$ denote a generalized eigenfunction of $H_\omega$ corresponding to energy $E$ and satisfying \eqref{eq:polyBounds}. After normalizing, we may assume that $u(\zeta) = 1$ for some choice of $\zeta \in \{0,1\}$. Define
\begin{equation} \label{eq:baseDef}
K
=
K(N, \zeta)
:=
\left\lceil \frac{1}{\e} \max\left(N, 2 \log^2(|\zeta|+1)\right) \right\rceil,
\end{equation}
where $N \in \Z_+$ is sufficiently large. More specifically, we take $N \ge N_0$, where $N_0 = N_0(\omega,\e) := \max \set{\widetilde n_0(\omega,\e), \widetilde n_1(\omega,\e),\widetilde n_2(\omega,\e)}$,  $\widetilde n_0$ comes from Proposition~\ref{prop:meanLDTshifted_sigma}, $\widetilde n_1$ is from Corollary~\ref{cor:LEandGreenEst_general}, and $\widetilde n_2$ is chosen so that $\omega \in \widetilde\Omega^{(N)} := \Omega_+(\e) \setminus\CD_N(\e)$ whenever $N \ge \widetilde n_2(\omega,\e)$.  Of course, right now, $\zeta \in\{0,1\}$, so $K \sim \e^{-1} N$. We keep up with the dependence on $\zeta$ to facilitate the ``second pass'' through the argument with $\zeta$ chosen to be a location at which $|u|$ is maximized.

Let us begin by first proving the following claim.
\begin{claim} \label{cl:grnfcts}
There exist integers $a_i,b_i$, $i =1,2$ such that
\begin{align}
\label{eq:main:a1range}
-K^9
\leq
a_1
& \leq
-K^3 +1 \\
\label{eq:main:a2range}
0
\leq a_2 & \leq K^9
\end{align}
and $b_i \in \{a_i + K^3-2, a_i+K^3-1, a_i+K^3\}$ such that
\begin{equation}\label{eq:main:ClaimGreenDecay}
\big|G_{T^\zeta\omega,\Lambda_i}^E(j,k)\big|
\leq
\exp\big(-|j-k|L(E) + C_0\e K^{3}\big),
\end{equation}
for any $j,k \in\Lambda_i := [a_i,b_i)$, where $C_0$ is a constant that depends only on $\widetilde\mu$.
\end{claim}

\begin{proof}
The claim follows by using Corollary~\ref{cor:LEandGreenEst_general} to estimate the Green function, \eqref{eq:transfertodet} to exchange the characteristic polynomial for a norm of a transfer matrix, and Proposition~\ref{prop:meanLDTshifted_sigma} to find a suitable starting point for said transfer matrix.

More precisely, we begin by applying Proposition~\ref{prop:meanLDTshifted_sigma} with $n=K^3$ twice: once with $\zeta_0 = \zeta$ and once with $\zeta_0 = \zeta-K^9$. With $\zeta_0 = \zeta$, we get
\[
\Big|L(E)-K^{-6}\sum_{s=0}^{K^{6}-1}
F_{K^3}(T^{\zeta + s K^3  }\omega)\Big|
<
\e;
\]
with $\zeta_0 = \zeta -K^9$, we have
\[
\Big|L(E)-K^{-6}\sum_{s=0}^{K^{6}-1}
F_{K^3}(T^{\zeta - K^9 + s K^3  }\omega)\Big|
<
\e.
\]
Note that the second invocation of Proposition~\ref{prop:meanLDTshifted_sigma} requires $K^3 \ge \log^{2/3}(|\zeta_0 - K^9|+1)$, which may be obtained by a direct computation.
	
Thus, a straightforward argument by contradiction yields $s_i$ satisfying $-K^9 \leq s_1 \leq -K^3$, $0 \leq s_2 \leq K^9 -K^3$, and
\begin{equation} \label{eq:largenorm_centered}
L(E)-\frac{1}{K^3}\log\|M_{K^3}^E(T^{\zeta + s_i}\omega)\|
<
\e.
\end{equation}
Since the norm of a $2\times 2$ matrix can be dominated by four times its greatest entry, this yields
\begin{equation}\label{eq:main:LEdetCompare_centered}
L(E) - \frac{1}{K^3}\log|\det[H_{T^{\zeta + a_i}\omega, k_i} - E]|\
<
2\e
\end{equation}
for some choice of $a_i \in \{s_i, s_i+1\}$ and $k_i \in \{K^3,K^3-1,K^3-2\}$ as long as $N_0$  is large enough that $N_0^{-3}\log 4 \leq \e$. Put $b_i = a_i + k_i$ and $\Lambda_i = [a_i,b_i)$.

Now, combining \eqref{eq:greenEst_K_general} with \eqref{eq:main:LEdetCompare_centered}, we obtain
\begin{align*}
\left|G_{T^\zeta\omega,\Lambda_i}^E(j,k)\right|
& =
\left|G_{T^{\zeta + a_i}\omega,k_i}^E(j,k)\right| \\
& \leq
\frac{\exp\left( (k_i - |j-k|)L(E) + C_0\e k_i \right)}{|\det[H_{T^{\zeta + a_i}\omega,k_i} - E]|} \\
& \leq
\frac{\exp\left((K^3 - |j-k|)L(E) + C_0\e K^3\right)}{\exp\left(K^3(L(E) - 2\e)\right)} \\
& =
\exp\left(-|j-k|L(E) + (C_0+2)\e K^3\right)
\end{align*}
for all $j,k \in \Lambda_i$. Note here we may need to uniformly enlarge $N_0$ to ensure $\e k_i \geq \max(\widetilde n_1, 2 \log^2(|\zeta+a_i|+1))$.
\end{proof}

\begin{claim}\label{claim:2}
Recall $u$ satisfies $|u(n)| \leq C_u(1+|n|)$ and $|u(\zeta)| = 1$. Then with $\Lambda_i$ as in Claim~\ref{cl:grnfcts}, put
\begin{equation*}
\ell_i
=
\left\lfloor \frac{a_i + b_i}{2} \right\rfloor.
\end{equation*}
Then we have
\begin{equation*}
|u(\zeta + \ell_i)|\leq e^{-2K^{2}},
\quad
i = 1,2.
\end{equation*}
whenever $N > N_0$ is large enough. Here, the size of $N_0$ depends on $\omega, \e$, and $u$ (with the dependence on $u$ entering solely through $C_u$).
\end{claim}

\begin{proof}
From \eqref{eq:deveGreen}, we obtain
\begin{align*}
|u(\zeta+\ell_i)|
& \leq
\big|G_{\omega,\Lambda_i}^E(\ell_i,a_i)\big|\big|u(\zeta+a_i-1)| + \big|G_{\omega,\Lambda_i}^E(\ell_i, b_i-1)\big|\big|u(\zeta + b_i)\big|
\end{align*}
Recall that $\log^2(|\zeta|+1) \leq K$ which implies $|\zeta| \leq e^{\sqrt K}$. This yields
\begin{align*}
|u(\zeta + \ell_i)|
&\leq C_u(K^9+e^{\sqrt K})e^{C_0\varepsilon K^3}
 (e^{-|\ell_i-a_i|L(E)}
+ e^{-|\ell_i-(b_i-1)|L(E)})\\
&\leq
2 C_u(K^9+e^{\sqrt K})\exp\Big(-L(E) K^{3}/3 + C_0\e K^{3}\Big).
\end{align*}
Then, for $N$ large enough (depending only on $\omega, \varepsilon, u$) we can estimate the last line by $e^{-2K^2}$, which concludes the proof of the claim.
\end{proof}
\begin{remark}\label{r:N0:indep:u}
 Claim~\eqref{claim:2} is the only place that requires the dependence of $N_0$ on $u$, and the dependence comes from the $u$-dependent constant $C_u$. In particular, if it is known that $u$ is normalized eigenvector with $\max |u(n)| = 1$, then the largeness of $N$ only depends on $(\omega,\varepsilon)$. In the argument after this remark, the largeness of $N_0$ will be independent of $u$. In particular, whenever we say ``enlarge $N_0$ if necessary'', the reader may verify that such an enlargement may be performed in a $u$-independent fashion.
\end{remark}
Now we use $|u(\zeta)| = 1$ and \eqref{eq:deveGreen} with $a = \zeta+\ell_1+1$ and $b = \zeta+\ell_2-1$ to get
\begin{align*}
1
& =
|u(\zeta)| \\
&\leq
|G_{\omega,[a,b]}^E(\zeta,a)| |u(\zeta+\ell_1)| + |G_{\omega,[a,b]}^E(\zeta,b)| |u(\zeta+\ell_2)|\\
& \leq
\left(G_{\omega,[a,b]}^E(\zeta,a)|+|G_{\omega,[a,b]}^E(\zeta,b)|\right)e^{-2K^2}.
\end{align*}
From this, we deduce
\begin{equation*}
\left\| G_{\omega,[a,b]}^E \right\|
\geq
\left\| G_{\omega,[a,b]}^E \delta_\zeta \right\|
\geq
\frac{\sqrt{2}}{2}e^{2K^2}
\geq
e^{K^2}.
\end{equation*}
Thus, \eqref{eq:doubleRes:GreenLB_centered} follows; since $0\le -a,b \leq K^9$, we can in turn conclude that expression \eqref{eq:doubleRes:FNupperBound_centered} fails for every $K^{10} \leq r \leq \overline K$ and $m=K,2K$; that is, we have
\begin{equation}\label{eq:main:FNlowerBd_centered}
\frac{1}{m}\log \|M_m^E(T^{\zeta+r}\omega)\|
>
L(E)-\e
\end{equation}
whenever  $K^{10} \leq r \leq \overline{K} \text{ and } m \in \{K,2K\}$.

Now use \eqref{eq:main:FNlowerBd_centered} to apply the Avalanche principle. Concretely, choose $n \in \Z_+$ with $K^{10} \le n \leq K^{-1}\overline{K} - K^9$, define
\[
A^{(j)}
:=
M_{K}^E(T^{\zeta + K^{10}+(j-1)K} \omega),
\quad 1 \le j \le n.
\]
With $\lambda := \exp(K(L(E) - \e))$, \eqref{eq:main:FNlowerBd_centered} gives
\[
\|A^{(j)}\|
\geq
\lambda
\geq
n
\]
for all $j$, where the second inequality holds as long as $N_0$ is sufficiently large. Since $K \ge \widetilde n_1$ and $K \geq \log^2(|\zeta| + |\overline{K}|+1)$ (enlarge $N_0$ if necessary), we may use \eqref{eq:FiniteLyapUpperEstimate_general} to obtain
\[
\|A^{(j)}\|
\le
\exp\big( K(L(E)+2\e)\big),\; 1\le j\le n.
\]
Thus, \eqref{eq:main:FNlowerBd_centered} implies
\begin{align*}
&\left| \log\|A^{(j+1)}\| + \log\|A^{(j)}\| - \log\|A^{(j+1)} A^{(j)}\| \right|\\
& <
2\lbase(L(E) + 2\e) - 2\lbase(L(E) - \e)\\
& =
6\lbase\e \\
& \leq
\frac{1}{2}\log\lambda,
\end{align*}
where the final inequality needs $\e$ to be sufficiently small; it is easy to see that this smallness condition depends only on $\widetilde \mu$ (through $\gamma$). Thus, taking $\hat N = nK$ and $r_0 =  K^{10}$, we have $\hat N \in [ K^{11}, \overline{K} - K^{10}]$ and the Avalanche Principle (Lemma~\ref{l.avlanche-principle}) yields
\begin{align*}
\log\|M_{\hat N}(T^{\zeta+ r_0}\omega)\|
& =
\log\|A^{(n)} \cdots A^{(1)}\| \\
& \geq
\sum_{j=1}^{n-1} \log\|A^{(j+1)}A^{(j)}\|
- \sum_{j=2}^{n-1} \log\|A^{(j)}\|
- C \frac{n}{\lambda} \\
& \geq
(n-1)2K(L(E) - \e) - (n-2)K(L(E) + 2\e) - C \\
& \geq
\hat N (L(E) - 5\e)
\end{align*}
by choosing $N_0$ large.

Putting this together, we can control $\|M_\ell^E(T^\zeta\omega)\|$ for general $K^{11} + K^{10} \leq \ell \leq \bar K$ by interpolation. In particular, writing $\ell = nK + p$ with $0 \leq p < K$ and $n\ge K^{10} + K^9$, we have
\begin{equation} \label{eq:largeLE:on:largescale}
\begin{split}
\|M_\ell^E(T^\zeta \omega)\|
& \geq
\frac{\|M_{\ell-K^{10}}(T^{\zeta + K^{10}}\omega)  \|}{\| M_{K^{10}}(T^\zeta \omega) \|} \\
& \geq
\Gamma^{-K^{10}-p} \|M_{nK-K^{10}}(T^{\zeta + K^{10}}\omega)\| \\
& \geq
\Gamma^{-K^{10} - p} e^{(nK -K^{10})(L(E) -5\e) } \\
& \geq e^{\ell (L(E) - 6\e)},
\end{split}
\end{equation}
as long as $N$ is sufficiently large. Since the intervals $[K^{11}+K^{10},\bar K]$ cover all sufficiently large integers, the foregoing estimates yield
\begin{equation} \label{eq:lowerbouldLE}
\liminf_{n \to \infty} \frac{1}{n} \log\|M_n^E(T^\zeta \omega)\|
\geq
L(E) - 6\e.
\end{equation}
Since $\zeta \in \{0,1\}$ and \eqref{eq:lowerbouldLE} holds for all $\e > 0$, we obtain \eqref{eq:main:FNLB}, as desired.
\end{proof}

\begin{proof}[Proof of Theorem~\ref{t:sule} {\rm(}SULE{\rm)}]
Let $\omega \in \Omega_*$, $E \in \CG(H_\omega)$, and $\d > 0$ be given. By Theorem~\ref{t:LyapConvGenEvals}, the associated eigenvector $u$ decays exponentially, and hence (after normalization), we may define the center of localization $\bar{\zeta}$ via $u(\bar\zeta) = \|u\|_\infty =  1$. There is an unimportant ambiguity here, namely, that $|u|$ can obtain its maximum value multiple times. However, since $u \in \ell^2$, it may only do so finitely many times, and it does not matter which of those occurrences we use for $\bar \zeta$. Then, define $K = K(N,\bar\zeta)$ as in \eqref{eq:baseDef}, and suppose $N$ is large.

Now, fix $\e > 0$ small; it will be apparent that how small depends only on $\widetilde\mu$ and $\d$. Then, running the proof of Theorem~\ref{t:LyapConvGenEvals} with $\zeta$ replaced by $\bar \zeta$, we obtain everything from Claim~\ref{cl:grnfcts} to \eqref{eq:largeLE:on:largescale} (with $\bar\zeta$ replacing $\zeta$). However, this time, we get $N \ge N_0(\omega,\e)$, that is, $N_0$ no longer depends on $u$. Then, for any $\ell$ with $K^{11} + K^{10} \le \ell \le \overline{K}$, we get
\[
\frac{1}{\ell}\log\| M_\ell^E(T^{\bar\zeta}\omega) \|
\geq
L(E) - 6 \e.
\]
In particular, combining this with \eqref{eq:transfertodet} and \eqref{eq:greenEst_K_general} implies that there exist $\bar\zeta' \in \{\bar \zeta,\bar\zeta + 1\}$ and $\ell' \in \{\ell, \ell-1,\ell-2\}$ so that
\begin{equation*}
\vabs{G^E_{T^{\bar \zeta'}\omega,\ell'} (j,k)} \leq \exp\left(-|j-k|L(E) + (C_0 + 6)\e \ell' \right)
\end{equation*}
for any $j,k \in [0,\ell')$. Pick $n\in[\frac{1}{4}\ell',\frac{1}{2}(\ell'-1)]$ (notice that $\ell' - n \ge n$). Then, \eqref{eq:deveGreen} and the normalization $\|u\|_{\infty} = 1$ yield
\begin{align*}
|u(\bar\zeta' + n)|
& \leq
|G^E_{T^{\bar\zeta'}\omega, \ell'}(0,n)| + |G^E_{T^{\bar\zeta' }\omega,  \ell'}(n,\ell'-1)| \\
& \leq
\exp\left( -L(E) n + (C_0+6) \e \ell' \right) + \exp\left(  -L(E) (\ell'-n) + (C_0+6) \e \ell'  \right)\\
& \leq
2 \exp\left( -L(E) n + 4(C_0+6)\e n \right) \\
& \leq
e^{-(1-\delta)L(E)n},
\end{align*}
where the last line needs $\e$ sufficiently small, dependent on the choice of $\delta > 0$. Shifting things back to $\bar\zeta$ if necessary, we get
\begin{equation}\label{eq:ExpDecayOfEigenFuncAtSigma}
\left\lvert u(\bar \zeta + n) \right\rvert
\leq
C e^{-(1-\delta)L(E)n}
\text{ for every }
n\in \left[\frac{1}{4}(K^{11} + K^{10}), \frac{1}{2} \overline{K-3}\right]
\end{equation}
These intervals are overlapping for $N_0$ large enough (and here, the largeness does not depend on any parameters), so we conclude that \eqref{eq:ExpDecayOfEigenFuncAtSigma} holds true for any
\[
n
\geq
\frac{\varepsilon^{-11}}{4}\max(N_0, 2\log^2(|\bar \zeta|+1))^{11},
\]
where $N_0$ depends only on $\omega$ (recall that $\e$ depends on $\delta$). For $0 \leq n \leq \frac{\varepsilon^{-11}}{4}\max(N_0, 2\log^2(|\bar \zeta|+1))^{11}$, we may estimate $|u(\bar \zeta + n)|$ trivially and adjust constants accordingly:
\begin{align*}
|u(\zeta + n)|
& \leq
e^{(1-\delta) L(E) \frac{\varepsilon^{-11}}{4}\max(N_0,2\log^2(|\bar \zeta|+1))^{11}} e^{-(1-\delta)L(E)n} \\
& \leq
C_{\omega,\delta} e^{C_\delta\log^{22}(|\bar\zeta|+1)}e^{-(1-\delta)L(E)n}
\end{align*}
where $C_{\omega,\delta}$ depends on $\omega, \delta$ and $C_\delta$ depends only on $\delta$. This proves the estimates in \eqref{sule} for all $n \geq 0$. To deal with $n < 0$, simply use reflection-invariance of $\Omega_*$.
\end{proof}

We can use the version of SULE from Theorem~\ref{t:sule} to prove the following version of almost-sure exponential dynamical localization. Recall that $\gamma$ denotes the global uniform lower bound on the Lyapunov exponent, compare \eqref{eq:UPLE}, and that it coincides with  the minimum of the Lyapunov exponent on the almost sure spectrum $\Sigma$, compare \eqref{eq:UPLEspec}.

\begin{theorem}[exponential dynamical localization] \label{t:DL}
For any  $\omega \in \Omega_*$, $\epsilon > 0$, and $0 < \beta < \gamma$, there is a constant $\widetilde C = \widetilde C_{\omega,\beta,\epsilon} > 0$ such that
\begin{equation*}
\sup\limits_{t \in \R}
|\sprod{\delta_n}{e^{-itH_\omega}\delta_m}|
\leq
\widetilde C e^{\epsilon |m|} e^{-\beta|n-m|}
\end{equation*}
for all $m,n \in \Z$.
\end{theorem}

For the remainder of the present section, we fix $\omega \in \Omega_*$ and leave the dependence of various quantities on $\omega$ implicit. Let $\set{u_\ell : \ell \in \Z}$ denote an enumeration of the normalized eigenvectors of $H_\omega$. By Theorem~\ref{t:sule}, each $u_\ell$ satisfies \eqref{sule} for a suitable choice of localization center $\zeta_\ell \in \Z$. For the remainder of this section, fix an arbitrary $\beta \in (0, \gamma)$. We need the following proposition which concerns the distribution of centers of localization.

\begin{prop}\label{DistCentLoc}
There exists $L_0\in \Z_+$ large enough that
$$
\#\{\ell : |\zeta_\ell| \leq L\}
\leq
L^2
$$
for all $L \geq L_0$.
\end{prop}

\begin{proof}
Let  $\CU_L := \{\ell \in \Z: |\zeta_\ell| \leq L\}$. When $L$ is sufficiently large, \eqref{sule}, yields
\[
|u_\ell(n)|
\leq
e^{2\beta L}e^{-\beta |n|}
\leq
e^{-\frac12 \beta|n|} \mbox{ whenever } \ell \in \CU_L \text{ and } |n| \ge 4L.
\]
This in turn implies
\begin{equation} \label{eq:smallcenter2}
\sqrt{ \sum_{|n|\ge 4L} |u_\ell(n)|^2 }
\leq
e^{-\beta L}
\end{equation}
whenever $L$ is sufficiently large. Let $u_{\ell,L}\in \R^{8L+1}$ be $P_{[-4L,4L]}u_\ell$. Using \eqref{eq:smallcenter2} and the fact that $\{u_\ell : \ell \in \CU_{L}\}$ is an orthonormal set, a direct computation shows that
\begin{equation} \label{eq:truncatedDotProducts}
\left|\langle u_{\ell,L}, u_{\ell',L}\rangle\right|
\begin{cases}
\ge 1 - e^{-\beta L}, & \ell = \ell',\\
\le 3e^{-\beta L},   & \ell \neq \ell'
\end{cases}
\end{equation}
whenever $L$ is sufficiently large. Now, consider the Gram matrix associated to $\{u_{\ell,L} : \ell \in \CU_L\}$, that is, the matrix $M$ having entries
\[
M_{\ell,\ell'}
=
\langle u_{\ell,L}, u_{\ell',L}\rangle,
\quad
\ell,\ell' \in \CU_L.
\]
For sufficiently large $L$, \eqref{eq:truncatedDotProducts} implies that $M$ is strictly diagonally dominant, hence invertible. In particular, $\{u_{\ell,L}: \ell \in \CU_{L}\}$ is a linearly independent set. Since these vectors are elements of $\R^{8L+1}$, it follows that $\#\CU_L \leq 8L+1$. Since $8L+1 < L^2$ whenever $L \ge 9$, the proposition follows.
\end{proof}

Now we are ready to prove Theorem~\ref{t:DL}. In fact, it is well-known that SULE-type conditions imply almost-sure exponential dynamical localization. We supply the details to keep the paper self-contained, following the argument in \cite[Theorem~7.5]{dRJLS1996}.

\begin{proof}[Proof of Theorem~\ref{t:DL}]
Given $\beta$ and $\epsilon$, choose $\beta'$ with $\beta < \beta' < \gamma$ and $\beta' - \beta =: \eta < \epsilon$.

Expanding $\delta_m$ in the basis of eigenfunctions of $H_\omega$, we obtain
\[
|\langle{\delta_n,e^{-itH_\omega}\delta_m\rangle}|
\leq
\sum_{\ell \in \Z}|u_\ell(n)u_\ell(m)|.
\]
Then, using \eqref{sule} we get
\[
\sum_\ell |u_\ell(n)u_\ell(m)|
\leq
C_\omega^2\sum_{\ell} e^{C_2\log^{22}(|\zeta_\ell| + 1)} e^{-\beta'(|n - \zeta_\ell| + |m-\zeta_\ell|)}
\]
Next, we have
\begin{align*}
e^{-\beta'(|n-\zeta_\ell| + |m-\zeta_\ell|)}
& \leq
e^{-\eta|\zeta_\ell|} e^{\eta |m|}e^{-(\beta' - \eta)|m-n|} \\
& \leq
e^{-\eta|\zeta_\ell|} e^{\epsilon |m|}e^{-\beta|m-n|}
\end{align*}
by the triangle inequality. Consequently,
\[
\sup_{t\in\R} \left| \langle{\delta_n, e^{-itH_\omega} \delta_0 \rangle} \right|
\leq
C_\omega^2\sum_\ell e^{C_2\log^{22}(|\zeta_\ell|+1)}  e^{-\eta |\zeta_\ell|} e^{\epsilon |m|} e^{-\beta|n-m|}.
\]
But then, by Proposition~\ref{DistCentLoc}, we have
\[
A_0
:=
\sum_\ell e^{C_2 \log^{22}(|\zeta_\ell| + 1) - \eta |\zeta_\ell|}
=
\sum_{L\ge 0}\ \sum_{|\zeta_\ell| = L} e^{C_2\log^{22}(|\zeta_\ell| + 1) - \eta|\zeta_\ell|}
<
\infty.
\]
Thus, we obtain
\[
\sup_{t\in\R}|\langle{\delta_n, e^{-itH_\omega}\delta_0\rangle}|
\leq
\widetilde C e^{\epsilon|m|}e^{-\beta |n-m|}
\]
with $\widetilde C = A_0 C_\omega^2$.
\end{proof}

\section{Localization for CMV Matrices with Random Verblunsky Coefficients}\label{sec:CMV}

In the present section, we will describe how to prove spectral and dynamical localization for CMV matrices with i.i.d.\ random Verblunsky coefficients. The overall outline of the proof is identical to the Schr\"odinger case. We will describe carefully the places where the proofs differ.

Let $\D\subset\C$ be the open unit disk. Suppose that our probability space $(\CA, \widetilde \mu)$ consists of a compact set of complex numbers in $\D$; that is, we assume henceforth that
\[
\CA
=
\supp\widetilde\mu
\subset
\D
\]
and that $\CA$ is compact. As in the Schr\"odinger case, we assume that $\# \CA\ge 2$ to avoid trivialities. Let $(\Omega,\mu)=(\CA^\Z, \widetilde\mu^\Z)$ and $T:\Omega\to\Omega$ denote the left shift. As before, the function $\alpha:\Omega\to \D$ given by $\a(\omega) := \omega_0$ can be used to generate Verblunsky coefficients via
\[
\alpha_\omega(n)
:=
\alpha(T^n\omega)
=
\omega_n,
\quad
n \in \Z,
\]
and we can (and do) view $\alpha_\omega = \{\alpha_\omega(n)\}_{n \in \Z}$ as a sequence of i.i.d.\ random variables on $\D$ with common distribution $\widetilde\mu$. For each $\omega \in \Omega$, let $\CC_\omega$ denote the CMV matrix associated with the Verblunsky coefficients $\{\alpha_\omega(n)\}_{n=0}^\infty$, and let $\CE_\omega$ be the extended CMV matrix associated with the coefficient sequence $\alpha_\omega$. Our immediate goal is to prove Theorems~\ref{t:ALCMV} and \ref{t:ALCMVhalfline}. We will discuss the proof of Theorem~\ref{t:ALCMV} in detail and then comment at the end of the section on the necessary changes for the proof of Theorem~\ref{t:ALCMVhalfline}.

Theorem~\ref{t:ALCMV} will be a corollary of the following theorem from which we may see the difference between Schr\"odinger operators and CMV matrices. In short, there is an exceptional set $\CD \subset \pD$ containing no more than three spectral parameters at which the hypotheses of F\"urstenberg's theorem may fail; thus, we work on compact arcs away from $\CD$, and we may exhaust $\pD \setminus \CD$ by countably many such arcs.

In the theorems below, $L(z)$ denotes the Lyapunov exponent for the operator family $\{\CE_\omega\}$, which shall be defined presently.

\begin{theorem}[Anderson localization for CMV matrices]  \label{t:ALCMV2}
With $\alpha_\omega$ as above, there exists a finite set $\CD$ with $\#\CD\le 3$ such that the following holds true. For any compact interval $\CI\subset \pD \setminus \CD$, there is a full-measure set $\Omega_\CI \subset \Omega$ such that $\CE_\omega$ has pure point spectrum on $\CI$ for every $\omega \in \Omega_\CI$, and the eigenfunctions of $\CE_\omega$ corresponding to any eigenvalue $z \in \CI$ decay exponentially. Moreover, the rate of decay is exactly $L(z)$.
\end{theorem}

Using the same two-pass approach as in Section~\ref{sec:localization}, we can prove a CMV version of SULE.

\begin{theorem}[SULE for CMV matrices] \label{t:CMVSULE}
Let $\CI$ and $\Omega_\CI$ be as in Theorem~\ref{t:ALCMV2}, and suppose $\omega \in \Omega_\CI$. For every $\delta>0$, there exist constants $C_\delta, C_{\omega,\delta}$ such that for every eigenfunction of $u$ of $\CE_\omega$ having eigenvalue $z \in \CI$, one has
\[
|u(\bar \zeta + n)|
\leq
C_{\omega,\delta} \|u\|_\infty e^{C_\delta \log^{22}(|\bar \zeta| + 1)} e^{-(1-\d)L(z)|n|}
\]
for all $n \in \Z$ and some $\bar \zeta = \bar \zeta(u)$.
\end{theorem}

Following the same arguments that led to Theorem~\ref{t:DL}, Theorem~\ref{t:CMVSULE} implies a version of dynamical localization for CMV matrices.

\begin{theorem}[dynamical localization for CMV matrices] \label{t:DLCMV}
Let $\CI$ and $\Omega_\CI$ be as in Theorem~\ref{t:ALCMV2}. For any $\omega \in \Omega_\CI$, $\epsilon > 0$, and any $\beta > 0$ with
\[
\beta
<
\gamma(\CI)
:=
\min_{z \in \CI} L(z),
\]
there is a constant $\widetilde C = \widetilde C_{\omega,\beta,\epsilon} > 0$ such that
\[
\sup_{k \in \Z}
\left| \left\langle\delta_m, \CE_\omega^k P_{\CI,\omega}  \delta_n \right\rangle\right|
\leq
\widetilde C e^{\epsilon|m|} e^{-\beta |n-m|},
\]
for all $m,n \in \Z$, where $P_{\CI,\omega} $ denotes the spectral projection of $\CE_\omega$ to the interval $\CI$. \end{theorem}

Assume for the moment Theorem~\ref{t:ALCMV2} holds true.

\begin{proof}[Proof of Theorem~\ref{t:ALCMV}]
For $z \in \pD$, define
$$
\Omega_z
:=
\{\omega \in \Omega : z \mbox{ is not an eigenvalue of } \CE_\omega\}.
$$
By a standard argument, $\mu(\Omega_z) = 1$ for each $z\in\partial \D$ (for example, the arguments of \cite{Pastur1980} can easily be modified to the CMV setting). For each integer $n\ge 2$, let $\Omega_{n}$ be the full measure set obtained from Theorem~\ref{t:ALCMV2} for
$$
\CI_n:=\partial\D\setminus \left(\bigcup\limits_{z\in\CD} \set{ ze^{i\t} : - 1/n < \t < 1/n} \right).
$$
Then, take
$$
\Omega_*
:=
\left(\bigcap\limits_{n\ge 2}\Omega_n\right)
\cap
\left(\bigcap\limits_{z\in\CD}\Omega_z\right).
$$
Clearly, $\mu(\Omega_*) = 1$. Moreover, for each $\omega \in \Omega_*$, $\CE_\omega$ has pure point spectrum on $\pD\setminus \CD$, exponentially decaying eigenfunctions for all eigenvalues $z \in \pD \setminus \CD$, and $\CD$ contains no eigenvalue of $\CE_\omega$. Hence $\CE_\omega$ exhibits Anderson localization for each $\omega\in\Omega_*$.
\end{proof}

Theorem~\ref{t:ALCMV2} and ~\ref{t:DLCMV} hold true for $\CC_\omega$ as well and the proofs are nearly identical. Once one has the half-line analog of Theorem~\ref{t:ALCMV2} in hand, the proof of Theorem~\ref{t:ALCMVhalfline} is exactly the same as the proof of Theorem~\ref{t:ALCMV}. We will focus on the proof of Theorem~\ref{t:ALCMV2} in the remainder of the present section and point out the differences between $\CE_\omega$ and $\CC_\omega$ in Remark~\ref{r:CMVhalfline} at the end of the paper.

\subsection{Uniform Positivity, Continuity, and LDT of the Lyapunov exponent}\label{s:SzegoLE}
For each $\alpha \in \CA$ and $z \in \pD$, we define the corresponding \emph{Szeg\H{o} transfer matrices} by
\begin{equation*}
S^z(\a)
:=
\frac{1}{\rho_\a}
\begin{bmatrix}
z & - \bar\a \\
-\a z & 1
\end{bmatrix},
\quad
M^z(\a)
:=
z^{-1/2} S^z(\a)
=
\frac{1}{\rho_\a}
\begin{bmatrix}
\sqrt{z}  &  -\frac{\bar{\alpha}}{\sqrt{z}} \\
-\alpha \sqrt{z} & \frac{1}{\sqrt{z}}
\end{bmatrix},
\end{equation*}
where $\rho_\a := \sqrt{1 - |\alpha| ^{2}}$. For concreteness, we choose the branch of $\sqrt{\cdot\,}$ defined by
\[
\sqrt{e^{i\t}} = e^{i\t/2},
\quad
-\pi < \t \leq \pi.
\]
Notice that $\|S^z\| = \|M^z\|$ for all $z \in \pD$.

We note that $M^z(\a) \in \SU(1,1)$ for every $z \in \pD$ and $\a \in \D$, where $\SU(1,1)$ is defined by
\begin{equation} \label{eq:SU11:conjdef}
\SU(1,1)
:=
Q\cdot \SL(2,\R)\cdot Q^*
=
\{A\in\C^{2\times 2}: A=QBQ^*, \mbox{ for some } B\in\SL(2,\R)\},
\end{equation}	
and
$$
Q
:=
\frac{-1}{1+i}
\begin{bmatrix}
1& -i\\
1& i
\end{bmatrix}\in \mathbb{U}(2).
$$
Equivalently, $\SU(1,1)$ consists of all $2\times 2$ unimodular matrices that preserve the standard quadratic form of signature $(1,1)$, that is, $\SU(1,1) = \set{A \in \SL(2,\C) : A^*JA = J}$, with $J = \vec e_1 \vec e_1^\top - \vec e_2 \vec e_2^\top$. We will freely use facts about the group $\SU(1,1)$ throughout this section; the interested reader is referred to \cite[Section~10.4]{S2} for a thorough account.
For $\omega\in\Omega$, we define an $\SU(1,1)$-cocycle via $M^z(\omega) = M^z(\omega_0)$ and
$$
M^z_n(\omega) = M^z(T^{n-1}\omega)\cdots M^z(\omega) = M^z(\omega_{n-1})\cdots M^z(\omega_0)
$$
for $n \in \Z_+$, as before. The Lyapunov exponent of the cocycle is then given by
$$
L(z)
=
\lim\limits_{n\to\infty} \frac1n\int_{\Omega} \! \log\|M^z_n(\omega)\| \, \dd \mu(\omega).
$$

Our first main goal is to obtain positivity and continuity of $L$ and use those  characteristics to deduce a suitable uniform LDT. We will deduce positivity and continuity by appealing to the machinery of Section~\ref{s:positivity_continuity_LE} and using that $\SU(1,1)$ is unitarily conjugate to $\SL(2,\R)$, as in \eqref{eq:SU11:conjdef}. The M\"obius transformation induced by $Q$ maps the upper half plane to the unit disk and sends the real line to the unit circle. Thus, in view of \eqref{eq:SU11:conjdef}, the M\"obius transformation induced by any element $\mathrm{SU}(1,1)$ preserves the unit circle and unit disk just like $\mathrm{SL}(2,\R)$ preserves the real line and the upper half-plane. Henceforth, we use $\bar A$ to denote the M\"obius transformation induced by an element $A \in \GL(2,\C)$.

Thus to verify all the necessary conditions, we may treat $\mathrm{SU}(1,1)$ matrices just as $\mathrm{SL}(2,\R)$ matrices. More concretely, condition (i) of Theorem~\ref{furst} and the contraction property may be verified directly (and these properties are clearly invariant under conjugation by $Q$). Since $\bar Q$ maps the real line to the unit circle, we will say that a subgroup $G \subset \SU(1,1)$ is \emph{strongly irreducible} if there is no finite subset $\CF \subset \pD$ such that $\bar B(\CF) = \CF$ for all $B \in G$; in particular, $G \subset \SU(1,1)$ is strongly irreducible in this sense if and only if $Q^* G Q$ is a strongly irreducible subgroup of $\SL(2,\R)$. In fact, once we know that $G$ is noncompact, we only need to verify condition (ii') as stated in the remark following Theorem~\ref{furst}. If $\nu$ is supported in $\SU(1,1)$, the appropriate version of condition (ii') is the statement: there is no $\CF\subset\pD$ with cardinality $1$ or $2$ such that $\bar B(\CF) = \CF$ for every $B \in G_{\nu}$.

In view of the foregoing discussion, our goal is to show that $G_{\nu_z}$ is non-compact and contracting for every $z\in\pD$, and to find a finite set $\CD = \CD(\CA)$ such that $G_{\nu_z}$ is strongly irreducible (as a subgroup of $\SU(1,1)$) for every $z\in \pD \setminus \CD$ (here, $\nu_z = M^z_* \widetilde\mu$ as before). Then, Theorem~\ref{furst} ensures that $L(z) > 0$ for all $z \notin \CD$. Moreover, since $G_{\nu_z}$ is contracting, $L(z)$ is continuous on $\pD$. Consequently, we obtain $L(z) \geq \gamma > 0$ for all $z$ in any compact interval $\CI \subset \pD \setminus \CD$. From there, the strong irreducibility condition and the contraction property ensure the uniform LDT (uniform over $z\in\CI$). Then, H\"older continuity of $L(z)$ on $\CI$ follows from the uniform positivity and uniform LDT of $L(z)$ and the fact that $\SU(1,1)$ is conjugate to $\SL(2,\R)$ (which ensures the applicability of the Avalanche Principle).

First, the following proposition follows from \cite[Lemma~10.4.14]{S2}.

\begin{prop}\label{p:hyperMz}
 If $[M^z(\a),M^z(\b)] := M^z(\a) M^z(\b) - M^z(\b) M^z(\a) \neq 0$, then the subgroup generated by $\{M^z(\a), M^z(\b)\}$ contains a non-elliptic element.
\end{prop}

\begin{proof}
If one of $M^z(\a)$ or $M^z(\b)$ is non-elliptic, there is nothing to do; otherwise, both are elliptic, in which case one may apply \cite[Lemma~10.4.14]{S2} to deduce the existence of a hyperbolic element in the subgroup of $\SU(1,1)$ that they generate.
\end{proof}

We first verify condition~(i) of Theorem~\ref{furst} and the contraction property.

\begin{prop}
For every $z \in \pD$, the group $G_{\nu_z}$ is noncompact and contracting.
\end{prop}

\begin{proof}
Let $z \in \pD$ be given, denote $G = G_{\nu_z}$, and let us note that it suffices to find a non-elliptic element $A \in G$ to obtain both noncompactness and the contraction property; more specifically, if $A \in \SU(1,1)$ is hyperbolic or parabolic, then $\|A^n\|$ becomes unbounded as $n\to\infty$ and $\|A^n \|^{-1} A^n$ converges to a rank-one operator. There are two cases to consider:
\medskip

\noindent \textbf{Case 1: \boldmath $z=1$.} Since $\CA$ contains at least two points, choose $\a \neq 0$ in $\CA$. Then,
\[
M^z(\a)
=
\frac{1}{\rho_\a}
\begin{bmatrix}
1 & -\bar\a \\
-\a & 1
\end{bmatrix}
=:
A
\]
satisfies $\mathrm{tr}(A) = 2\rho_\alpha^{-1} > 2$, so $A$ is a hyperbolic element of $\SU(1,1)$.
\medskip

\noindent \textbf{Case 2: \boldmath $z \neq 1$.} In this case, choose $\a\neq\b$ in $\CA$. Then, one can check that
\begin{equation} \label{eq:Mzcommutator}
M^z(\a) M^z(\b)
-
M^z(\b) M^z(\a)
=
\frac{1}{\rho_\a \rho_\b}
\begin{bmatrix}
\bar{\a}\b - \a\bar{\b} & (\bar{\a} - \bar{\b})(1-z^{-1})  \\
(\a-\b)(1-z) & \a\bar{\beta} - \bar{\a}\beta
\end{bmatrix}
\neq
0.
\end{equation}
Therefore, $G_{\nu_z}$ contains a non-elliptic element by Proposition~\ref{p:hyperMz}, as desired.
\end{proof}

Next we note that strong irreducibility condtion (ii') essentially follows from \cite[Theorem 10.4.15]{S2}.

\begin{prop} \label{p:CMVfurst:ii'}
Suppose $\#\CA\ge 2$, then condition (ii') fails at most at a finite set $\CD\subset \pD$ with $\#\CD\le 3$.
\end{prop}

From the proof of Proposition~\ref{p:CMVfurst:ii'} as in \cite{S2}, we see that the set $\CD$ arises from very particular geometric degeneracies, and it is in fact empty for many choices of a compact set $\CA \subset \D$. More precisely, in view of \cite[Theorem~10.4.19]{S2} we see that $\CD$ is empty as soon as the following two conditions are met:
\begin{itemize}
\item $\CA$ is not contained in a single circle or line that intersects $\pD$ orthogonally.
\item The set
\[
\set{\frac{|\Im(\bar \alpha \beta)|}{|\alpha-\beta|} : \alpha \neq \beta \text{ and } \alpha, \beta \in \CA}
\]
contains at least two elements.
\end{itemize}
In particular, one can simply take $\CI = \pD$ when these conditions hold true. Note this in particular implies that a full version of exponential dynamical localization in Theorem~\ref{t:DLCMV}, i.e. there is no need to add $P_\CI(\CE_\omega)$. One may find more details regarding the first condition in \cite{S2}. To meet the second condition, $\CA$ must contain at least three non-colinear points, and, if $\#\CA = 3$, then the incenter of the triangle\footnote{That is, the point of intersection of the angle bisectors of that triangle.} with vertices at the points of $\CA$ must be distinct from 0.

\subsection{Estimating Transfer Matrices and Green Functions}\label{s:GreenEstCMV}
From now on, we will focus on an arbitrary fixed compact interval $\CI\subset\pD\setminus\CD$, on which a uniform LDT for $L(z)$ shall hold. Following the arguments of previous sections, one can show the following analog of Proposition~\ref{prop:meanLDTshifted_sigma}. The proof is nearly identical; one need only replace $E \in \hat\Sigma$ by $z \in \CI$ and make small cosmetic modifications.

\begin{prop} \label{prop:CMVmeanLDTshifted}
For any $0 < \e < 1$, there exists a subset $\Omega_+ = \Omega_+(\e) \subset \Omega$ of full $\mu$-measure such that the following statement holds true. For every $\e \in (0,1)$ and every $\omega \in \Omega_+(\e)$, there exists $\widetilde n_0 = \widetilde n_0(\omega,\e)$ such that
\begin{equation} \label{eq:CMVpartsum}
\left|L(z) - \frac{1}{n^{2}}\sum_{s=0}^{n^{2}-1} \frac{1}{n} \log \left\| M^z_n(T^{sn +\zeta} \omega) \right\| \right|
<
\e
\end{equation}
for every $\zeta\in\Z$, every $n \geq \max\{\widetilde n_0, \log^{\frac23}(|\zeta|+1)\}$, and every $z \in \CI$.
\end{prop}

Now we consider the results in Section~\ref{sec:green} for finite-volume truncations of CMV matrices. We will exploit the perspective on CMV Green functions developed in \cite{Krueg2013IMRN}. Here it will helpful to use the following factorization of $\CE_\omega$. Writing
\[
\Theta(\alpha)
=
\begin{bmatrix} \overline\alpha & \sqrt{1-|\alpha|^2} \\ \sqrt{1-|\alpha|^2} & - \alpha
\end{bmatrix},
\quad
\CL_\omega
=
\bigoplus_{j\in\Z}\Theta(\alpha_{2j}(\omega)),
\quad
\CM_\omega
=
\bigoplus_{j\in\Z} \Theta(\alpha_{2j+1}(\omega)),
\]
where $\Theta(\alpha_n(\omega))$ acts on coordinates $n$ and $n+1$, one can confirm that $\CL_\omega$ and $\CM_\omega$ are unitary and that $\CE_\omega = \CL_\omega\CM_\omega$. Of course, $\CE_\omega \psi = z\psi$ if and only if $(z\CL_\omega^* - \CM_\omega)\psi = 0$. Given $\tau_1,\tau_2 \in \overline{\D}$ and an interval $\Lambda = [a,b] \subseteq \Z$, we define $\CE_{\omega}^{\tau_1,\tau_2}$ to be the CMV matrix whose Verblunsky coefficients coincide with those of $\CE_\omega$, except $\alpha_{a-1} = \tau_1$ and $\alpha_b = \tau_2$. We then define
\[
\CE_{\omega,\Lambda}^{\tau_1,\tau_2}
=
P_\Lambda \CE_{\omega}^{\tau_1,\tau_2} P_\Lambda^*.
\]
One can verify that $\CE_{\omega,\Lambda}^{\tau_1,\tau_2}$ is unitary whenever $\tau_1,\tau_2 \in \pD$. Following the convention of \cite{Krueg2013IMRN}, we use $\bullet$ to indicate that the corresponding Verblunsky coefficient is unaltered, e.g.
\[
\CE_\omega^{\bullet,\tau_2} = \CE_{\omega}^{\alpha_{a-1},\tau_2},
\quad
\CE_{\omega}^{\tau_1,\bullet} = \CE_{\omega}^{\tau_1,\alpha_b}.
\]
As before, we abbreviate $\CE_{\omega,N}^{\tau_1,\tau_2} = \CE_{\omega,[0,N)}^{\tau_1,\tau_2}$. The truncations $\CL_{\omega,\Lambda}^{\tau_1,\tau_2}$ and $\CM_{\omega,\Lambda}^{\tau_1,\tau_2}$ with $\tau_j \in \overline{\D} \cup \{\bullet\}$ are defined similarly. Then, we define the polynomials
\[
\varphi_{\omega,\Lambda}^{\tau_1,\tau_2}(z)
=
\det\left[z - \CE_{\omega,\Lambda}^{\tau_1,\tau_2}\right],
\quad
z \in \C, \; \tau_j \in \overline{\D}\cup\{\bullet\}.
\]
The associated finite-volume Green functions are defined by
\[
G_{\omega,\Lambda}^{\tau_1,\tau_2}(z)
=
\left(z\left[\CL_{\omega,\Lambda}^{\tau_1,\tau_2}\right]^* - \CM_{\omega,\Lambda}^{\tau_1,\tau_2} \right)^{-1}
\]
and
\[
G_{\omega,\Lambda}^{\tau_1,\tau_2}(j,k;z)
=
\langle \delta_j, G_{\omega,\Lambda}^{\tau_1,\tau_2}(z)
 \delta_k \rangle,
\quad
j,k \in \Lambda.
\]
By \cite[Proposition~3.8]{Krueg2013IMRN}, for $\tau_j \in \pD$, these objects are related via
\[
\left| G_{\omega,\Lambda}^{\tau_1,\tau_2}(j,k;z) \right|
=
\frac{1}{\rho_j\rho_k}
\left| \frac{\varphi_{\omega,[a,j-1]}^{\tau_1,\bullet} (z) \varphi_{\omega,[k+1,b]}^{\bullet,\tau_2}(z) }{\vspace{.5cm}\varphi_{\omega,[a,b]}^{\tau_1,\tau_2}(z)} \right|,
\quad
 a \leq j \leq k \leq b,
\]
which furnishes the CMV analog of \eqref{eq:greenstodet}. To connect Green functions and transfer matrices \`a la \eqref{eq:transfertodet}, we use \cite[Corollary~3.11 and Lemma~3.12]{Krueg2013IMRN}, which gives
\[
\left| \varphi_{\omega,[a,j-1]}^{\tau_1,\bullet} (z) \right|
\leq
\sqrt{2} \|S^z_j(T^a \omega)\|,
\quad
\left| \varphi_{\omega,[k+1,b]}^{\bullet,\tau_2}(z) \right|
\leq
\sqrt{2}\|S^z_{b-k}(T^{k+1}\omega)\|.
\]
as well as
\begin{equation} \label{eq:szegoTM2Dets}
\varphi_{\omega,[a,b]}^{\tau_1,\tau_2}(z)
=
\left\langle
\begin{bmatrix} 1 \\ -\overline{\tau_2} \end{bmatrix}, \;
S^z_{b-a}(T^a\omega)
\begin{bmatrix} 1 \\ \overline{\tau_1} \end{bmatrix}
\right\rangle
\end{equation}
Consequently, the minor twist here is the following: when we want to deduce good Green function estimates from largeness of the transfer matrices, we wiggle the boundary condition instead of the interval. That is, if $\|S^z_n(\omega)\|$ is large, then \eqref{eq:szegoTM2Dets} implies that $\varphi^{\tau_1,\tau_2}_{\omega,[0,n)}(z)$ is large for at least one choice of $\tau_j \in \{\pm 1\}$.

\begin{prop} \label{prop:LEandGreenEst:CMV}
For any $0 < \e < 1$ and $\omega \in \Omega_+(\e)$, there exists $\widetilde n_1 = \widetilde n_1(\omega,\e)$ large enough so that the following statements hold true.
\begin{equation} \label{eq:normupbound}
\frac{1}{N} \log\|M_N^z(T^{\zeta}\omega) \|
\leq
L(z) + 2\e
\end{equation}
for all $\zeta\in\Z$, $n \ge \max\{\widetilde n_1, \log^2(|\zeta|+1)\}$, $z\in\CI$. Moreover,  the following holds with $C_0$ dependent only on $\widetilde \mu$:
\begin{equation} \label{eq:greenfnest1:CMV}
\left| G^{\tau_1,\tau_2}_{T^{\zeta}\omega,N}(j,k;z) \right|
\leq
C_0 \frac{e^{(N-|j-k|)L(z) + C_0 \e N}}{\left| \varphi_{T^\zeta\omega,[0,N)}^{\tau_1,\tau_2}(z) \right|}
\end{equation}
for all $j,k \in [0,n)$, $\zeta\in \Z$, $n\ge \frac{1}{\e}\max\{\widetilde n_1, 2\log^2(|\zeta|+1)\}$, $\tau_j \in \{\pm 1\}$, $z \in \CI \setminus \sigma(\CE_{T^{\zeta}\omega , N}^{\tau_1,\tau_2})$.
\end{prop}

The proof of Proposition~\ref{prop:LEandGreenEst:CMV} is entirely analogous to that of Corollary~\ref{cor:LEandGreenEst_general}. The $C_0$ in front comes from the factor $\rho_j^{-1}\rho_k^{-1}$ and hence only depends on the support of $\widetilde \mu$.

\subsection{Proof of Anderson Localization}\label{s:ProofALCMV}

The statement and proof of the CMV analog of Proposition~\ref{prop:doubleRes} (elimination of double resonances) is almost identical to the Schr\"odinger operator setting, with the twist that we need to allow for four different boundary conditions. Concretely, one defines $\CD_N(\varepsilon)$ to denote those $\omega$ such that
\[
|F_m(T^{\zeta+r}\omega,z)| \leq L(z) - \varepsilon
\]
and
\[
\| G_{T^\zeta\omega,[-N_1,N_2]}^{\tau_1,\tau_2}(z)\|
\geq
e^{K^2}
\]
for some choice of $m$, $\zeta$, $r$, $K$, $N_j$ as before, some $z \in \CI$, and some choice of $\tau_j \in \{\pm 1\}$.

Once we eliminate double resonances, we are ready to prove our Theorem~\ref{t:ALCMV2}. We need an appropriate version of Schnol's theorem to guarantee that spectrally almost every $z \in \pD$ is a generalized eigenvalue of $\CE$, which is supplied by \cite{DFLY2}. So, as before, we may work with $\omega$ in a full-measure set and $\xi$ a linearly bounded generalized eigenfunction of $\CE_\omega$, normalized by $\xi_0 = 1$. We note the following difference in the proofs of Claims~1 and 2 in this setting.

 In the proof of Claim 1, the appropriate CMV analog of \eqref{eq:largenorm_centered} is the same as in the Schr\"odinger case. Then, by \eqref{eq:szegoTM2Dets}, one can choose $\tau_1,\tau_2 \in \{\pm 1\}$ so that
\[
\left| \varphi_{T^\zeta\omega,\Lambda_i}^{\tau_1,\tau_2}(z) \right|
\geq
e^{K^3(L(z) - 2\varepsilon)}.
\]
Combining this with \eqref{eq:greenfnest1:CMV}, we obtain the CMV version of \eqref{eq:main:ClaimGreenDecay}, i.e.
\begin{equation}\label{eq:main:ClaimGreenDecayCMV}
\big|G_{T^\zeta \omega,\Lambda_i}^{\tau_1,\tau_2}(j,k;z)\big|
\leq
\exp\big(-|j-k|L(z) + C\e K^{3}\big),
\end{equation}
for any $j,k \in\Lambda_i$ and this same choice of $\tau_i$, which is enough for our purpose.

Finally, in the proof of Claim~2, we need the CMV version of \eqref{eq:deveGreen}, which is supplied by \cite[Lemma~3.9]{Krueg2013IMRN}. Concretely, if $u$ is a solution of the difference equation $\CE u = zu$, define
\[
\widetilde\psi(a)
=
\begin{cases}
\left(z\overline{\tau_1}-\alpha_a\right)u(a) - \rho_a u(a+1)
& a \text{ is even,} \\
\left(z\alpha_a - \tau_1\right)u(a) + z\rho_a u(a+1)
& a \text{ is odd,}
\end{cases}
\]
and
\[
\widetilde\psi(b)
=
\begin{cases}
\left(z\overline{\tau_2}-\alpha_b\right)u(b) - \rho_b u(b-1)
& b \text{ is even,} \\
\left(z\alpha_b - \tau_2\right)u(b) + z\rho_{b-1} u (b-1)
& b \text{ is odd.}
\end{cases}
\]
Then, we have
\begin{equation}\label{eq:CMV:deveGreen}
u(n)
=
G_{[a,b]}^{\tau_1,\tau_2}(n,a;z)\widetilde\psi(a)
+ G_{[a,b]}^{\tau_1,\tau_2}(n,b;z)\widetilde\psi(b)
\end{equation}
for $a < n < b$, and of course, $\widetilde\psi$ is linearly bounded whenever $u$ is. Then, combining \eqref{eq:CMV:deveGreen} with \eqref{eq:main:ClaimGreenDecayCMV} and following the proof of Claim~2, we obtain for all $k$ near the center of $\Lambda_i$,
\begin{equation} \label{eq:smallelement}
|\xi_k|\le Ce^{-2K^2}.
\end{equation}
Now, writing $\Lambda_i = [a_i,b_i]$, take $a = \lfloor (a_1+b_1)/2\rfloor$, $b = \lfloor (a_2+b_2)/2\rfloor$, and consider $G_{T^\zeta,[a,b]}$. By the foregoing arguments, $\xi_\ell$ satisfy \eqref{eq:smallelement} for $\ell$ near or at $a,b$, and $\xi$ is normalized so that $\xi_0=1$, we obtain
$$
\|G^z_{T^\zeta\omega, [a,b]}\|\ge |G^z_{T^\zeta,[a,b]}(0,\ell)|\ge ce^{K^2}
$$
for some $\ell$ near or at $a$ or $b$.

Then, the remainder of the proof of Theorem~\ref{t:ALCMV2} is almost identical to the corresponding arguments for the Schr\"odinger case. In particular, we get that
\[
\lim_{|n| \to\infty} \frac{1}{|n|} \log\|M^z_n(\omega)\| = L(z) > 0.
\]
To relate this back to quantitative exponential decay estimates for the generalized eigenfunctions of $\CE_\omega$, we need to use the Gesztesy--Zinchenko (GZ) transfer matrices \cite{GZ2006}, not the Szeg\H{o} transfer matrices. However, this is not a big deal, because there is a simple connection between these matrices \cite{DFLY2}. Thus, we conclude that the generalized eigenfunctions of $\CE_\omega$ are exponentially decaying at $\pm\infty$ (at the rate $L(z)$).

Finally, using the CMV tools, we can make a second pass through the argument and prove Theorem~\ref{t:CMVSULE} (SULE), which in turn implies Theorem~\ref{t:DLCMV} (Dynamical Localization).

\begin{remark}\label{r:CMVhalfline}
For the proof of the half-line version of Theorem~\ref{t:ALCMV2} and Theorem~\ref{t:DLCMV}, we note that $\CC_{[a,b]} = \CE_{[a,b]}$ whenever $1 \le a \le b$; when $a = 0$, we have $\CC_{[0,b]} = \CE_{[0,b]}$ with the modification $\alpha_{-1} = -1$. Moreover, the Szeg\H{o} transfer matrices $M^z_n(\omega)$ remain the same as long as $n\ge 0$. We can then obtain the half-line analogs of the results in Sections~\ref{s:SzegoLE} and \ref{s:GreenEstCMV} simply by following the arguments in those sections and suitably restricting the domains of $n$ and $[a,b]$.

The main difference is in the statement of elimination of double resonances. Here, we need to change $G^z_{T^\zeta\omega, [-N_1,N_2]}$ to $G^z_{T^\zeta\omega, [0,N_2]}$ in one of the conditions. Note that $G^z_{[0,N)}$ now refers to the Green function for $z - \CC_{[0,N)}$. After that, the remainder of the proof follows the same argument as before.
\end{remark}

\section*{Acknowledgments}

The main idea of the new proof of the LDT in Section~\ref{sec:LDT} was communicated to D.D.\ and Z.Z.\ by Artur Avila while they were visiting IMPA, Rio de Janeiro. They would like to thank Artur Avila for sharing his idea and IMPA for the hospitality. We are grateful to Anton Gorodetski and G\"unter Stolz for useful input.

\end{document}